\documentclass{amsart}
\pdfoutput=1
\usepackage[foot]{amsaddr}
\usepackage{amssymb}

\usepackage{graphicx}

\usepackage{multicol}
\usepackage[bottom]{footmisc}
\usepackage{bm}
\usepackage{cite}

\usepackage[margin=2cm]{geometry}
\usepackage{lipsum}

\usepackage{float}
\usepackage{subfig}

\usepackage{mathtools}
\usepackage{graphicx}
\usepackage{multirow}
\usepackage[mathscr]{euscript}
\usepackage{enumitem}
\usepackage{changepage}
\usepackage{bbm}
\newtheorem{theorem}{Theorem}[section]
\newtheorem{lemma}[theorem]{Lemma}

\theoremstyle{definition}

\theoremstyle{remark}

\usepackage{color}

\newtheoremstyle{examplestyle}
   {}{}{}{}{\bfseries}{.}{.5em}{{\thmname{#1 }}{\thmnumber{#2}}{\thmnote{ (#3)}}}
\theoremstyle{examplestyle}

\newcommand{\noi}{\mathcal{W}} 
\newcommand{\kr}{K} 
\newcommand{\gau}{\mathcal{G}} 

\usepackage[linesnumbered,ruled]{algorithm2e}

\makeindex             

\begin{document}

\title[  ]{A Finite-Volume Method for Fluctuating Dynamical Density Functional Theory}

\author{Antonio Russo\textsuperscript{\dag}}
\address[Antonio Russo]{Department of Chemical Engineering, Imperial College London, SW7 2AZ, UK}
\curraddr{}
\thanks{\noindent\textsuperscript{\dag} A. Russo and S. P. Perez are co-first and contributed equally to this work. S. P. Perez is the corresponding author.}

\author{Sergio P. Perez\textsuperscript{\dag}}
\address[Sergio P. Perez]{Departments of Chemical Engineering and Mathematics, Imperial College London, SW7 2AZ, UK}
\curraddr{}
\email{sergio.perez15@imperial.ac.uk}
\thanks{}

\author{Miguel A. Dur\'an-Olivencia}
\address[Miguel A. Dur\'an-Olivencia]{Department of Chemical Engineering, Imperial College London, SW7 2AZ, UK}
\curraddr{}
\thanks{}

\author{Peter Yatsyshin}
\address[Peter Yatsyshin]{Department of Chemical Engineering, Imperial College London, SW7 2AZ, UK}
\curraddr{}
\thanks{}

\author{Jos\'e A. Carrillo}
\address[Jos\'e A. Carrillo]{Mathematical Institute, University of Oxford, Oxford OX2 6GG, United Kingdom}
\curraddr{}
\thanks{}

\author{Serafim Kalliadasis}
\address[Serafim Kalliadasis]{Department of Chemical Engineering, Imperial College London, SW7 2AZ, UK}
\curraddr{}
\thanks{}


%
%
%
%

\begin{abstract}
We introduce a finite-volume numerical scheme for solving stochastic
gradient-flow equations.
Such equations are of crucial importance within the framework of fluctuating hydrodynamics and dynamic density functional theory.
Our proposed scheme deals with general free-energy functionals, including,
for instance, external fields or interaction potentials.
This allows us to simulate a range of physical phenomena where thermal fluctuations play a crucial role, such as nucleation and other energy-barrier crossing transitions.
A positivity-preserving algorithm for the density is derived based on a
hybrid space discretization of the deterministic and the stochastic terms
and different implicit and explicit time integrators.
We show through numerous applications that not only our scheme is able to
accurately reproduce the statistical properties (structure factor and
correlations) of the physical system, but, because of the multiplicative
noise, it allows us to simulate energy barrier crossing dynamics, which
cannot be captured by mean-field approaches.

\end{abstract}


\maketitle

%
%
%
%

\section{Introduction}

The study of fluid dynamics encounters major challenges due to the inherently multiscale nature of fluids. Not surprisingly, fluid dynamics has been one of the main arenas of activity for numerical analysis and fluids are commonly studied via numerical simulations, either at molecular scale, by using molecular dynamics (MD) or Monte Carlo (MC) simulations; or at macro scale, by utilising deterministic models based on the conservation of fundamental quantities, namely mass, momentum and energy.
While atomistic simulations take into account thermal fluctuations, they come with an important drawback, the enormous computational cost of having to resolve at least three degrees of freedom per particle.
Despite drastic improvements in computational power over the last few decades, atomistic simulations are only applicable for small fluid volumes. There are also other challenges with such techniques, e.g. the use of a proper thermostat when running non-equilibrium simulations at constant temperature~\cite{russo2019macroscopic}.
On the contrary, the convenience of partial differential equations (PDEs),
such as continuity and Navier-Stokes, is enormous as they are amenable to
both analytical and numerical scrutiny, with numerical simulations being less
computationally expensive than MD-MC.
However, continuous models based upon PDEs cannot account for the stochastic nature observed in real systems.
Fortunately, there is still an approach which lives at the crossroad of mesoscale, namely fluctuating hydrodynamics (FH). Firstly proposed by Landau and Lifshitz~\cite{landau1980statistical}, FH is formulated in terms of stochastic PDEs which aim at extending Navier-Stokes equations to include thermal fluctuations.
FH can then be used to simulate systems undergoing energy-barrier crossing
transitions, such as nucleation, which are impossible to describe within the
mean-field approximation.

However, the FH formulation by Landau and Lifshitz~\cite{landau1980statistical} is phenomenological: they simply
included additive stochastic flux
terms in the Navier-Stokes equations -- we shall refer to these equations as
the Landau-Lifshitz-Navier-Stokes (LLNS) equations.
A remarkable effort has been made ever since trying to connect FH with MD from first principles~\cite{Bixon1969,Fox1970,Mashiyama1978,Kawasaki1994,Dean1996,chavanis2008hamiltonian,Espanol2015}.
Some of the most widely known attempts to formalise such a connection are the works of Kawasaki~\cite{Kawasaki1994} and Dean~\cite{Dean1996}. Theirs provide a formal derivation of the stochastic time-evolution equation for the "density" field of a system of Brownian particles. Nevertheless, their derivation ends up with a time-evolution equation for the microscopic density field, which is nothing but a re-writing of the Brownian equations by using Itô's lemma. For this reason, the Dean-Kawasaki equation has been actively criticised and seen more as a toy model whose derivation does not represent a proper proof of the FH equation.
Indeed this model cannot be employed to describe macroscopic quantities,
such as density and momentum fields which are obtained by ensemble averaging
the corresponding microscopic quantities \cite{Archer2004}, and thus remains
disconnected from the original Landau-Lifshitz theory. And it is this
disconnection that has led to the misconception that the Dean-Kawasaki model describes the evolution of macroscopic observables.

In a recent work~\cite{duran-olivencia2017}, a bottom-up derivation of the FH for a system of Brownian particles has been posed.
It provided a new first-principle formulation of
the governing equations for macroscopic observables in the framework of
classical dynamic density functional theory (DDFT).
It is worth mentioning that the field of DDFT has gained a lot of traction since the first phenomenological derivations proposed in \cite{evans1979nature,dieterich1990nonlinear}.
	Several rigorous derivations have been put forward including effects such as inertia, hydrodynamic interactions and orientation of particles.
	These derivations have been proposed both for the overdamped and inertial regimes, and we refer the reader to \cite{goddard2012general,goddard2012overdamped, lutsko2010recent,duran-olivencia2016,archer2009dynamical} for more details.
In particular, the formulation proposed in~\cite{duran-olivencia2017} allows for a rigorous and systematic derivation of FH but
also fluctuating DDFT (FDDFT) which includes the effects of thermal
fluctuations on the mean-field DDFT.
In that work, it is also shown how the
classical DDFT is the most-likely realisation of FDDFT, thus providing
closure to a long standing debate in the classical DFT community about the
inclusion of fluctuations in DFT. Also, the derivation in \cite{duran-olivencia2017} stays in tune with the original
intuitive treatment of Landau and Lifshitz and at the same time alleviates
the misconceptions with the Dean-Kawasaki model.
As a remark, it should be noticed that LLNS equations describe a full
	system of particles, while FDDFT governs the time-evolution of density and
	momentum fields of subcomponents of a system, e.g. of colloidal particles in
	a bath. Because of the momentum exchange between colloidal and bath
	particles, the total momentum in FDDFT for colloidal particles is not
	conserved, being affected by thermal fluctuations and friction exerted by the
	bath.
Let us also note here that both classical DFT and DDFT, embedded with either exact
or approximated models for the density-dependent Helmholtz free-energy
functional \cite{Marconi1999}, have been shown to be rather powerful in
the study of complex systems at the nano- and microscale
\cite{Goddard2012,goddard2012general,Donev2014}. Recent advances in DFT and DDFT have
extended its applicability to a wide spectrum of applications from nucleation
of colloids and macromolecules
\cite{lutsko2012dynamical,lutsko2013classical,DuranOlivencia2018} to fluids in
confined geometries \cite{Yatsyshin2015,Goddard2016,Nold2017} and wetting
phenomena \cite{Nold2014,Yatsyshin2015_1,Yatsyshin2018}. But also highly
non-uniform systems such as dense liquid droplets and solid
clusters~\cite{Lutsko2019}.

The FDDFT framework in Ref.~\cite{duran-olivencia2017}, derived for the
general case of arbitrarily shaped and thermalized particle, consists of two
stochastic PDEs for the number density $\rho$ (known also as particles state probability function) and velocity $\mathbf{v}$ fields:
	\begin{equation}
	\begin{cases}
	\partial_t \rho ( \mathbf{r}  ,t) + \boldsymbol{\nabla}_\mathbf{r} \cdot \left( \rho( \mathbf{r}  ,t) v( \mathbf{r}  ,t) \right)  = 0,& \\[7pt]
	\partial_t \left( m \rho ( \mathbf{r}  ,t) \mathbf{v}( \mathbf{r}  ,t) \right) + \boldsymbol{\nabla}_\mathbf{r} \cdot \left( m \rho ( \mathbf{r}  ,t) \mathbf{v}( \mathbf{r}  ,t) \otimes \mathbf{v}( \mathbf{r}  ,t) \right)+ \rho ( \mathbf{r}  ,t) \boldsymbol{\nabla}_\mathbf{r} \dfrac{\delta \mathcal{E} [\rho] }{\delta \rho( \mathbf{r}  ,t)}  &\\[7pt]
	\hspace{4cm}+ m \gamma \rho( \mathbf{r}  ,t) \mathbf{v}( \mathbf{r}  ,t) + \sqrt{ k_B T m \gamma \rho ( \mathbf{r}  ,t) } \boldsymbol{\noi} ( \mathbf{r}  ,t) = 0,
	\end{cases}
	\end{equation}
where $m$ is the mass of the particles, $\mathcal{E} [\rho] $ is the
density-dependent free-energy functional, $\gamma$ is a friction parameter
describing the interactions between the particles and the bath, $k_B$ is the
Boltzmann constant, $T$ is the temperature and $\boldsymbol{\noi}$ is a
vector of Gaussian stochastic processes delta-correlated in space and time,
i.e.
\begin{align}
\langle \boldsymbol{\noi} ( \mathbf{r} ,t) \rangle =&0, \\
\langle \boldsymbol{\noi}( \mathbf{r}  ,t),\boldsymbol{\noi}( \mathbf{r}' ,t') \rangle =& 2 \delta(t-t') \delta(\mathbf{r}-\mathbf{r}').
\end{align}
In the strong damping limit ($m^{-1} \gamma \rightarrow \infty$), the high
friction between the particles and the bath causes the characteristic time
scale of the momentum dynamics to be much shorter than the density one
\cite{goddard2012overdamped,duran-olivencia2017}.
Thus, as a first approximation, the contributions of the terms $
\boldsymbol{\nabla}_\mathbf{r} \cdot \left( m \rho \mathbf{v} \otimes
\mathbf{v} \right)$ and $\partial_t \left( m \rho \mathbf{v} \right) $ can
be neglected. As a result, one obtains the stochastic time-evolution equation
for the density field, referred to as overdamped FDDFT
\cite{duran-olivencia2017,Kruger2017}:
\begin{align}\label{eq:overdamped}
\partial_t \rho ( \mathbf{r}  ,t) =  \boldsymbol{\nabla}_\mathbf{r} \cdot \left( \left( m \gamma \right)^{-1} \rho( \mathbf{r}  ,t) \boldsymbol{\nabla}_\mathbf{r} \dfrac{\delta \mathcal{E}[\rho]}{\delta \rho( \mathbf{r}  ,t)} \right) + \boldsymbol{\nabla}_\mathbf{r} \cdot \left( \sqrt{ k_B T  \left( m \gamma \right)^{-1}  \rho( \mathbf{r}  ,t)   } \boldsymbol{\noi} ( \mathbf{r}  ,t) \right) .
\end{align}
Equation \eqref{eq:overdamped} may be seen as a stochastic version of the
gradient-flow equation previously studied, for instance, in Refs
\cite{villani2003topics,carrillo2003kinetic}.
As we later discuss in further detail, Eq.~\eqref{eq:overdamped} reduces to the stochastic diffusion equation \cite{kim2017} when considering a system of non-interacting particles (ideal gas), whose free energy would be $ \mathcal{E}[\rho] = \int \rho \left( \log \rho -1 \right)
d \mathbf{r}$.
However, the presence of a more general functional $ \mathcal{E}[\rho] $ allows
in principle to introduce non-linear diffusion, external force fields and
interparticle interactions.
It is also worth mentioning that Eq.~\eqref{eq:overdamped} is not well-posed due to the high irregularity originated in the stochastic fluxes and the multiplicity of the noise.
	This difficulty is typically overcome by introducing some sort of regularization, such as a finite-volume interpretation as employed here and in previous works \cite{donev2010}.
	One also needs to be careful about the cell size choice and possible nonphysical effects such as negative densities, which may arise from the Gaussian processes.
	In Sect.~\ref{sec:nummet} we propose a finite-volume interpretation of Eq.~\eqref{eq:overdamped} which correctly overcomes these issues.

Previous numerical methodologies for FH have been focused on the LLNS
equations for the density and momentum, and the energy equation for the
temperature if the systems are non-isothermal.
In comparison, the overdamped FDDFT allows us to obtain the density field solving a single equation with stochastic fluxes for isothermal systems.
One of the first works on this regard is by Garcia \textit{et al.}
\cite{garcia1987numerical}, where a simple finite-difference scheme to treat
the numerical fluxes of the SPDE is constructed.
Further works by Bell \textit{et al.}
\cite{bell2007numerical,bell2010computational} provide an explicit Eulerian
discretization of the LLNS equations combined with a third-order Runge-Kutta
method with the objective of adequately reproducing the fluctuations in
density, energy and momentum.
Donev and co-workers~\cite{donev2010} exploited the structure factor
(equilibrium fluctuation spectrum) to construct finite-volume schemes to
solve the LLNS which then allows one to study the accuracy for a given
discretization at long wavelengths.

	They also proposed a Petrov-Galerkin finite-element discretization of non-linear stochastic diffusion equations embedded with prototypical free-energy functionals, such as the Ginzburg-Landau free energy \cite{delaTorre2015}.
	However, in order to obtain analytical forms of the structure factors used to
	assess the performance of the scheme, the study focused on systems at
	equilibrium (i.e. without density discontinuities), at supercritical
	temperatures (to avoid phase transition phenomena), and without any external
	potential.
Similarly, methods to solve FH via staggered grids have been constructed~\cite{balboa2012staggered}.
Other works have proposed numerical schemes based on temporal
integrators that are implicit-explicit predictor-corrector \cite{delong2013temporal} or two-level leapfrog \cite{glotov2014new}.
Additionally, hybrid schemes have been developed to couple LLNS with MD
\cite{delgado2007embedding,de2007fluctuating,de2006multiscale} or with MC
\cite{williams2008algorithm,donev2010hybrid} simulations of complex fluid
systems.
Moreover, the LLNS have also been solved to tackle reactive multi-species fluid mixtures \cite{bhattacharjee2015fluctuating}.
Further works have developed numerical schemes for particular applications of
the overdamped FDDFT in Eq.~\eqref{eq:overdamped}.
Specifically, Refs \cite{kim2017,atzberger2010spatially} developed numerical
methods for reaction-diffusion equations obtained by adding appropriate
reaction terms to Eq.~\eqref{eq:overdamped} equipped with the ideal-gas
free-energy functional.

The works just mentioned have contributed to a better understanding of the effects of thermal fluctuations in complex fluid systems.
Nevertheless, an efficient and systematic numerical methodology to solve
Eq.~\eqref{eq:overdamped} equipped with a general free-energy
functional has not yet been developed.
Such a methodology would allow for the simulation and scrutiny of a wide
range of non-equilibrium phenomena which can be studied within the framework
of FDDFT.
Relevant examples of these physical phenomena include dynamic evolution of confined systems and energy-barrier crossing transitions, such as nucleation.

In this work we introduce a finite-volume method to solve general
stochastic gradient-flow equations with the structure of
Eq.~\eqref{eq:overdamped} for FDDFT.
The main advantages of finite-volume schemes are the conservation of the total mass of the system and the flexibility to simulate complex geometries.

	The main contributions of this work can be summarized as follows:
	\begin{itemize}
		\item To provide a space discretization scheme able to deal with fluctuations at discontinuous density profiles. We discretize the deterministic fluxes based on a hybrid approach which takes advantage of both central and upwind schemes.
		\item To overcome the commonplace challenge of preserving non-negative densities in the presence of noise, by adopting a Brownian bridge technique. Despite previous approaches employing artificial limiters \cite{kim2017}, our
		technique ensures density positivity without altering the Gaussian
		distribution of the stochastic field.
		\item To develop a methodology to simulate a family of free-energy functionals, modelling different physical systems.
		First, we study temporal and spatial correlations, and the structure factor of
		ideal gas at equilibrium, comparing the results of our finite-volume solver
		with both MD and theoretical results.
		Then, we examine the out-of-equilibrium evolution of an ideal gas in a
		double-well external potential.
		Subsequently, we simulate homogeneous nucleation kinetics of a fluid
		consisting of particles interacting through a Lennard-Jones (LJ)-like
		potential.
		Providing initial uniform densities corresponding to metastable vapour
		conditions, we study the phase-transition of the system and compare the
		results with the mean-field phase diagram.
		\item To implement and test families of implicit-explicit Euler and Milsten time integrators, together with a weak second-order Runge-Kutta scheme.
		\item To gain insights into the free-energy decay for stochastic gradient-flow equations (see for instance Figs. 8(d) and 10(b)). The decay of free energy is an important feature of deterministic gradient-flow equations. However, in stochastic gradient-flow equations, the free-energy decay is guaranteed only in the weak noise limit.
	\end{itemize}

In Sect.~\ref{general}, we present the model equation to simulate and outline
its main properties.
In Sect.~\ref{sec:nummet}, we discuss the numerical methodology of our
finite-volume scheme, including flux discretization, time integrators,
adaptive time step to preserve density positivity and boundary conditions.
Several applications to illustrate the validity of our methodology are presented in Sect.~\ref{numerical_applications_eq}.
Finally, a summary and conclusions are offered in Sect.~\ref{conclusions}.
%
%
%
%

\section{Governing equations and related properties}
\label{general}

Our starting point is the following general SPDE based on the overdamped
FDDFT in Eq.~\eqref{eq:overdamped} with $\gamma=1$ and $m=1$,
\begin{align}\label{eq:FDDFToverd}
\begin{cases}
\partial_t \rho (\bm r ;t ) = \boldsymbol{\nabla}_\mathbf{r} \cdot \left[\rho(\bm r ;t ) \boldsymbol{\nabla}_\mathbf{r} \dfrac{\delta \mathcal{E} [\rho] }{\delta \rho(\bm r ;t )}  \right]  + \boldsymbol{\nabla}_\mathbf{r} \cdot \left[ \sqrt{ \rho (\bm r ;t ) / \beta} \boldsymbol{\noi} (\bm r ;t ) \right]  \quad \mathbf{r} \in \mathbb{R}^d ,\ t>0,\\
\rho(\mathbf{r} ;0)=\rho_0(\mathbf{r}),
\end{cases}
\end{align}
where $\mathcal{E}[\rho]$ denotes the free energy of the system given by
\begin{align}\label{eq:freesimple}
\mathcal{E}[\rho]  = \int_{\mathbb{R}^d} f(\rho)d\mathbf{r} +\int_{\mathbb{R}^d}  V( \mathbf{r} ) \rho\, d \mathbf{r} + \frac{1}{2} \int_{\mathbb{R}^d} g \left(  \kr( \mathbf{r} ) * \rho ( \mathbf{r} ) \right) \rho ( \mathbf{r}) \,d \mathbf{r},
\end{align}
with $f(\rho)$ describing the dependency of the free energy
$\mathcal{E}[\rho]$ on the local density field $\rho(\mathbf{r})$, $V( \mathbf{r})$
accounting for the effects of external potentials, $g$ denoting a function
depending on the convolution of $\rho(\mathbf{r})$ with the symmetric kernel $\kr(
\mathbf{r} )$ accounting for the interparticle potential. For simplicity, we
introduce the constant $\beta$, defined as $\beta=(k_B T)^{-1}$.

The mean-field limit of Eq.~\eqref{eq:FDDFToverd} in which no stochastic flux
is present has received a great deal of attention in the context of gradient
flows.
As discussed in Ref.~\cite{duran-olivencia2017}, the most likely path, in the weak noise limit, minimizes the Lagrangian defined as $\mathcal{L}= \| \partial_t \rho -  \boldsymbol{\nabla}_{\bm{r}} \cdot \left(  \rho(\bm r; t) \boldsymbol{\nabla}_{\bm r}  \frac{\delta \mathcal{E}[\rho]}{\delta \rho}  \right) \|_{(\sigma \sigma^*)^{-1}}$, where $\sigma$ is the operator acting on the noise $\boldsymbol{\noi} (\bm r ;t )$.
Thus, the most-likely solution $\langle {\rho} \rangle (\bm r ;t ) $ satisfies
\begin{align}
\partial_t \langle {\rho} \rangle (\bm r ;t )  =  \boldsymbol{\nabla}_{\bm{r}} \cdot \left(  \langle {\rho} \rangle \boldsymbol{\nabla}_{\bm r}  \frac{\delta \mathcal{E}[\langle {\rho} \rangle]}{\delta \langle {\rho} \rangle}   \right).
\label{eq:deterministic}
\end{align}
Equation~\eqref{eq:deterministic} is a generalized diffusion equation,
which results in the heat equation if an ideal gas free energy is selected.
It has been widely employed not only in the framework of DDFT \cite{goddard2012general,goddard2012overdamped,Yatsyshin2015}, but also to model thin-liquid films stochastic dynamics\cite{Duran-Olivencia2019}.
It has the structure of a gradient flow in the Wasserstein metric
\cite{otto1996double,villani2003topics} with applications in a variety of
contexts such as granular media \cite{carrillo2003kinetic}, materials science
and biological swarming
\cite{carrillo2003kinetic,ReinaZimmer2015,barbaro2016phase}.
The fundamental property of Eq.~\eqref{eq:deterministic} is that the free energy \eqref{eq:freesimple} is minimized following the decay rate \cite{carrillo2003kinetic,Carrillo2006,Carrillo2019}
\begin{equation}
\frac{d}{dt}\mathcal{E}[\langle \rho\rangle ] =-\int_{\mathbb{R}^d} \langle \rho\rangle \left| \frac{\delta \mathcal{E}[ \langle \rho\rangle ]}{\delta \langle \rho\rangle }\right|^2 d \mathbf{r},
\label{eq:decay_rate}
\end{equation}
where the variation of the free energy $\mathcal{E}[\rho]$ with respect to the density $\rho$ in the case of \eqref{eq:freesimple} satisfies
\begin{equation}\label{eq:varfreesimple}
\frac{\delta \mathcal{E}[\rho]}{\delta \rho}= f'(\rho) + V( \mathbf{r} ) + \kr *(g'(\kr * \rho)\rho)+g(\kr * \rho).
\end{equation}
The decay rate in Eq.~\ref{eq:decay_rate} is not satisfied by the stochastic gradient flow in Eq.~\eqref{eq:FDDFToverd}, where occasional free-energy increase can take place during the dynamical evolution.
It is precisely these jumps that allow the system to overcome energy barriers leading
to phenomena such as phase transitions.

%
%
%
%

\subsection{Structure factor}
\label{Structure factor}

The structure factor is a quantity of interest in many fields, including FH
\cite{donev2010} as noted earlier and capillary wave theory
\cite{Paulus2008,Parry2016}.
As shown in previous works~\cite{donev2010,kim2017}, the structure factor
represents an important measure of the stochastic properties of the system
and it can be experimentally obtained.
Thus, the structure factor is a valuable quantity not only to study the stability of the numerical integrator, but also to compare different schemes, as it will be shown in Sect.\ref{sec:nummet}.
Here we derive an expression of the structure factor from the linearized
FDDFT.
If we consider a periodic domain of volume $V$, the spatial Fourier transform of the density is given by
\begin{align}\label{eq:st1}
\hat{\rho}_{\boldsymbol{\lambda} } = \frac{1}{V} \int_V \rho ( \mathbf{r} , t ) e^{-i \boldsymbol{\lambda} \cdot \mathbf{r} } d \mathbf{r}.
\end{align}
The structure factor is defined as the variance of the Fourier transform of the density fluctuations,
\begin{align}\label{eq:st2}
S ( \boldsymbol{\lambda} )= V \langle  \delta \hat{\rho}_{\boldsymbol{\lambda} }  \delta \hat{\rho}_{\boldsymbol{\lambda} }^*   \rangle ,
\end{align}
where $ \delta \hat{\rho}_{\boldsymbol{\lambda} }  =  \hat{\rho}_{ \boldsymbol{\lambda} }  - \langle  \hat{\rho}_{\boldsymbol{\lambda} } \rangle $, and $\hat{\rho}_{ \boldsymbol{\lambda} }^*$ denotes the complex conjugate of $\hat{\rho}_{ \boldsymbol{\lambda} }$.

For uniform systems, Eq.~\eqref{eq:FDDFToverd} can be formally linearized around its most-likely solution $ \langle {\rho} \rangle $ by means of the Central Limit Theorem, giving
\begin{align}
\partial_t \rho (\bm r; t) =  \boldsymbol{\nabla}_{\bm{r}} \cdot \left(  \rho(\bm r; t) \boldsymbol{\nabla}_{\bm r}  \frac{\delta \mathcal{E}[\rho]}{\delta \rho}  \right) +\sqrt{ \langle {\rho} \rangle  / \beta } \boldsymbol{\nabla}_{\bm{r}} \cdot \boldsymbol{ \mathcal{W} } (\bm r; t) .
\label{eq:linearized}
\end{align}
Taking the Fourier transform of the difference between Eq.~\eqref{eq:linearized} and Eq.~\eqref{eq:deterministic}, one obtains
\begin{align}
\begin{split}
\partial_t \  \delta \hat{ \rho} ( \boldsymbol{\lambda}) =
i \boldsymbol{\lambda} \cdot
\left\{ \mathcal{T} \left(  \rho(\bm r; t) \boldsymbol{\nabla}_{\bm r}  \frac{\delta \mathcal{E}[\rho]}{\delta \rho} \right) - \mathcal{T} \left(  \langle {\rho} \rangle \boldsymbol{\nabla}_{\bm r}  \frac{\delta \mathcal{E}[\langle {\rho} \rangle]}{\delta \langle {\rho} \rangle}   \right) \right\} +  i \boldsymbol{\lambda} \cdot \sqrt{ \langle {\rho} \rangle  / \beta  }  \hat{ \boldsymbol{ \mathcal{W} } } (\boldsymbol{\lambda}) .
\end{split}
\label{eq:fourier_rho}
\end{align}
where $\mathcal{T}$ denotes the Fourier transform.
If the free-energy functional terms in the Fourier space can be expanded at first order around their mean value as
\begin{align}
\mathcal{T}  \left(  \rho(\bm r; t) \boldsymbol{\nabla}_{\bm r}  \frac{\delta \mathcal{E}[\rho]}{\delta \rho}  \right)   \sim
\mathcal{T} \left(  \langle {\rho} \rangle \boldsymbol{\nabla}_{\bm r}  \frac{\delta \mathcal{E}[\langle {\rho} \rangle]}{\delta \langle {\rho} \rangle}   \right) +
\frac{ \partial \mathcal{T}  \left[  \rho(\bm r; t) \boldsymbol{\nabla}_{\bm r}  \frac{\delta \mathcal{E}[\rho]}{\delta \rho} \right]  }{ \partial \hat{ \rho}_{\boldsymbol{\lambda}} } \delta \hat{ \rho}_{\boldsymbol{\lambda}} +  \mathcal{O} (\delta \hat{ \rho}_{\boldsymbol{\lambda}} ),
\end{align}
then Eq.~\ref{eq:fourier_rho} yields
\begin{align}
\begin{split}
\partial_t \  \delta \hat{ \rho}_{\boldsymbol{\lambda}} =
i \boldsymbol{\lambda} \cdot
\frac{ \partial \mathcal{T}  \left[  \rho(\bm r; t) \boldsymbol{\nabla}_{\bm r}  \frac{\delta \mathcal{E}[\rho]}{\delta \rho}  \right]  }{ \partial \hat{ \rho}_{\boldsymbol{\lambda}} } \delta \hat{ \rho}_{\boldsymbol{\lambda}} +
i \boldsymbol{\lambda} \cdot \sqrt{ \langle {\rho} \rangle  / \beta   }  \hat{ \boldsymbol{ \mathcal{W} }  }(\boldsymbol{\lambda}) .
\end{split}
\end{align}
Since the above equation has the form of an Ornstein-Uhlenbeck process, the
structure factor can be computed as its variance:
\begin{align}\label{eq:structurefgeneral}
\begin{split}
S ( \boldsymbol{\lambda} )=
\frac{ 2 \left(  i \boldsymbol{\lambda} \sqrt{ \langle {\rho}  \rangle  / \beta } \right)^2}{2 i \boldsymbol{\lambda}
	\frac{ \partial \mathcal{T}  \left[  \rho(\bm r; t) \boldsymbol{\nabla}_{\bm r} \frac{\delta \mathcal{E}[\rho]}{\delta \rho} \right]  }{ \partial \hat{ \rho}_{\boldsymbol{\lambda}} } }=
\frac{ i \boldsymbol{\lambda}  \langle {\rho} \rangle / \beta }{
	\frac{ \partial \mathcal{T}  \left[  \rho(\bm r; t) \boldsymbol{\nabla}_{\bm r}  \frac{\delta \mathcal{E}[\rho]}{\delta \rho} \right]  }{ \partial \hat{ \rho}_{\boldsymbol{\lambda}} } }
\end{split}
\end{align}
For example, in the case of an ideal gas without external potential,
$\frac{\delta \mathcal{E}[\rho]}{\delta \rho} = \log \rho$, the structure
factor is given by the well-known expression~\cite{kim2017}:
\begin{align}\label{eq:structuref}
\begin{split}
S ( \boldsymbol{\lambda} )=
\frac{  i \boldsymbol{\lambda}  \langle {\rho}  \rangle  / \beta  }{
	\frac{ \partial \mathcal{T}  \left[  \rho(\bm r; t) \boldsymbol{\nabla}_{\bm r} \log \rho(\bm r, t) \right] }{ \partial \hat{ \rho}_{\boldsymbol{\lambda}} } }=
\frac{ i \boldsymbol{\lambda}  \langle {\rho}  \rangle  / \beta   }{
	\frac{ \partial \mathcal{T}  \left[ \boldsymbol{\nabla}_{\bm r}  \rho(\bm r, t) \right] }{ \partial \hat{ \rho}_{\boldsymbol{\lambda}} } }=
\frac{ i \boldsymbol{\lambda}  \langle {\rho} \rangle  / \beta  }{
	\frac{ \partial  \left[ i \boldsymbol{\lambda}  \hat{\rho} \right] }{ \partial \hat{ \rho} ( \boldsymbol{\lambda} ) } }=\langle {\rho}  \rangle  / \beta .
\end{split}
\end{align}

%
%
%
%

\section{Numerical methods}
\label{sec:nummet}

The one-dimensional (1D) version of Eq. \eqref{eq:FDDFToverd} can be written as
\begin{equation}\label{eq:FDDFTflux}
\partial_t\rho=\partial_x F_d (\rho) +\partial_x F_s (\rho, \noi),
\end{equation}
where $F_d$ and $F_s$ denote the deterministic and stochastic fluxes, respectively,
\begin{equation}\label{eq:fluxes}
F_d = \rho \partial_x \dfrac{\delta \mathcal{E}[\rho]}{\delta \rho},\quad F_s=\sqrt{ \rho/ \beta} \noi.
\end{equation}
The finite-volume formulation of Eq. \eqref{eq:FDDFTflux} is obtained by
dividing the domain into grid cells
$C_j=[x_{j-\frac{1}{2}},x_{j+\frac{1}{2}}]$, each one assumed to have the
same length $\Delta x=x_{j+1/2}-x_{j-1/2}$, and then approximating in each of
them the cell average of $\rho$ defined as
\begin{align}\label{eq:rhoaver}
\overline{\rho}_j(t)=\frac{1}{\Delta x }\int_{x_{j-1/2}}^{x_{j+1/2}} \rho(x,t) dx.
\end{align}
Subsequently, one has to integrate \eqref{eq:FDDFTflux} spatially over each cell and apply the Gauss divergence theorem, leading eventually to the semi-discrete equation for the temporal evolution of the cell average density,
\begin{equation}\label{eq:fvfluxes}
\frac{d \overline{\rho}_j}{d t}=\frac{F_{d,j+1/2}-F_{d,j-1/2}}{\Delta x}+\frac{F_{s,j+1/2}-F_{s,j-1/2}}{\Delta x},
\end{equation}
where $F_{d,j+1/2}$ and $F_{s,j+1/2}$ denote the deterministic and stochastic
fluxes \eqref{eq:fluxes} evaluated at the boundary $x_{j+1/2}$. The
separation of the physical flux into deterministic and stochastic parts has
been effectively applied in previous studies
\cite{bell2007numerical,donev2010}, noting though that some of them consider a single flux combining the deterministic and stochastic
terms~\cite{mohamed2013finite}.
Here we treat them separately.
We now proceed to develop in detail the methodology of
	our finite-volume scheme.

\subsection{Deterministic flux}

The deterministic flux is evaluated by employing a hybrid method which
adopts a central or upwind approximation depending on the relative local
total variation of the density. This is a classical technique in
deterministic fluid dynamics to construct high-resolution and
oscillation-free schemes \cite{sweby1984high}. On the one hand, central
high-order and non-diffusive schemes are applied wherever smooth gradients
of the density are found. On the other hand, a diffusive upwind scheme is
employed in those regions of the domain with density gradients, in order to
prevent the spurious oscillations from central high-order schemes.

Previous works in the field of FH~\cite{bell2007numerical,donev2010,kim2017}
approximate the deterministic flux with a simple second-order central
difference approach, even though high-order differences are also proposed but
not implemented \cite{donev2010}. Our motivation to adopt  a hybrid approach
is precisely aimed to avoid possible spurious oscillations. The previous
literature is mainly focused on FH with $f'(\rho)=\log \rho$ in
Eq.~\eqref{eq:varfreesimple}, resulting in a deterministic flux of the form
$\partial_x F_d(\rho)=\partial_{xx} \rho$. The treatment of this Laplacian
with a central approximation works well for the cases presented in the
literature, but as it is shown later in Fig.~\ref{fig:deterministic_flux}, it
can cause spurious oscillations for some solutions.

In the case of Eq.~\eqref{eq:FDDFTflux}, the stochastic flux leads to
non-smooth density. The proposed gradient scheme then compares the
local gradient in the density with the neighbouring gradients. When the local
gradient is large when compared to the neighbours', an upwind approximation is chosen. If not,
the central approximation prevails. As a result, our proposed hybrid scheme
for the deterministic flux satisfies
\begin{equation}\label{eq:centralupwind}
F_{d,j+1/2}=\left(1-\phi(r_{j+1/2}) \right) F_{d,j+1/2}^c+\phi(r_{j+1/2}) F_{d,j+1/2}^u,
\end{equation}
where $\phi(r_{j+1/2})$ is a flux limiter with a threshold parameter $k$, defined as
\begin{equation*}
\phi(r_{j+1/2})= \begin{cases}
0,\quad \text{if}\, r_{j+1/2}\leq k,\\
1,\quad \text{if}\, r_{j+1/2}>k,
\end{cases}
\end{equation*}
and $r_{j+1/2}$ is a quotient measuring the relative local variation of the density,
\begin{equation}
r_{j+1/2}=\frac{\left| \rho_{j+1}-\rho_j\right|}{\sum_{l=-w}^{w}\left| \rho_{l+1}-\rho_l\right|},
\end{equation}
with $w$ indicating the number of neighbouring cell used to compute the total variation.
A value $w=5$ is employed in the numerical experiments of this work, since it gives a good compromise between conservation of local information and effects of the fluctuations.

The threshold parameter $k$ plays a key role and has to be carefully
selected. When $k$ is small, the diffusive upwind scheme is chosen more
frequently, leading to diffusive behaviour which affects the structure factor
and the correlations. On the contrary, when $k$ is large, the central scheme
will be predominant, and spurious oscillations may be created.
Fig.~\ref{fig:deterministic_flux} provides a numerical example to choose an
adequate value for $k$.

Firstly, Figs.~\ref{fig:deterministic_flux}(a-b) are obtained by simulating
\eqref{eq:FDDFTflux} with a free energy satisfying $\delta \mathcal{E}/\delta
\rho=\log \rho+ 0.1 x$. The initial density profile has two discontinuities
as shown in Fig.~\ref{fig:deterministic_flux} (a). Under these conditions,
the numerical solution evolves as a diffusive travelling wave, but the two
discontinuities in the initial density trigger spurious oscillations. The
oscillations diminish by reducing $k$ (for $k=0$, which corresponds to only
upwind flux, the diffusion eliminates the oscillations). However, a low value
of $k$ critically dampens the variance, due to the diffusive nature of the
upwind flux, as it is noticed from Fig.~\ref{fig:deterministic_flux} (b).

Secondly, Fig.~\ref{fig:deterministic_flux} (c) is obtained from simulating
\eqref{eq:FDDFTflux} with a free energy satisfying $\delta \mathcal{E}/\delta
\rho=\log \rho$ and starting from an equilibrium density profile. For this
case, the theoretical value of the structure factor is known and is given by
\eqref{eq:structuref}, meaning that the dampening behaviour of the upwind
scheme could be directly evaluated from Fig.~\ref{fig:deterministic_flux}
(c). We notice again how the upwind scheme dampens the
statistical properties of the system due to the numerical diffusion. As a
result, an intermediate value of $k$ needs to be taken in order to find a
balance between both numerical flaws. The compromising value is chosen to be
$k=3$.

\begin{figure}[ht!]
	\begin{center}
		\subfloat[]{\protect\protect\includegraphics[scale=0.6]{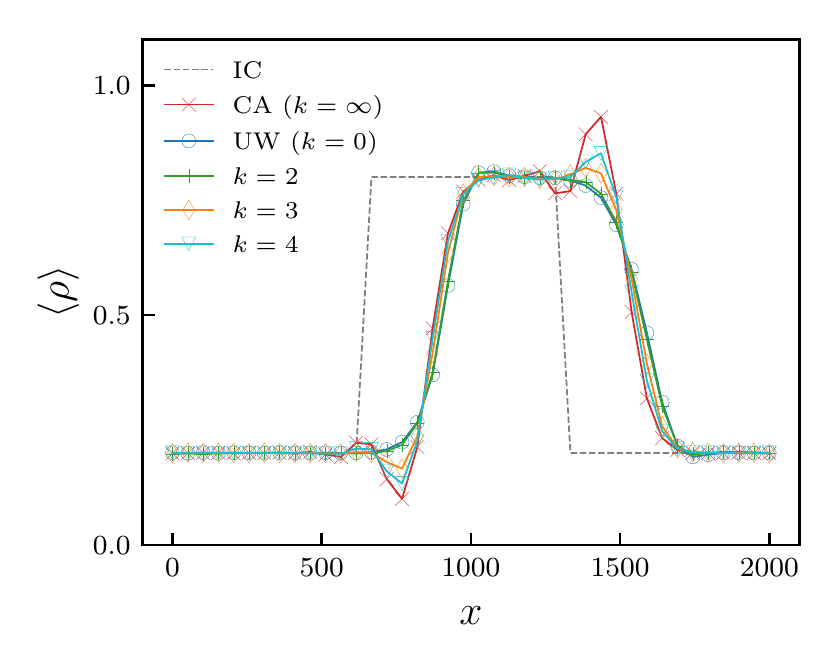}
		}
		\subfloat[]{\protect\protect\includegraphics[scale=0.6]{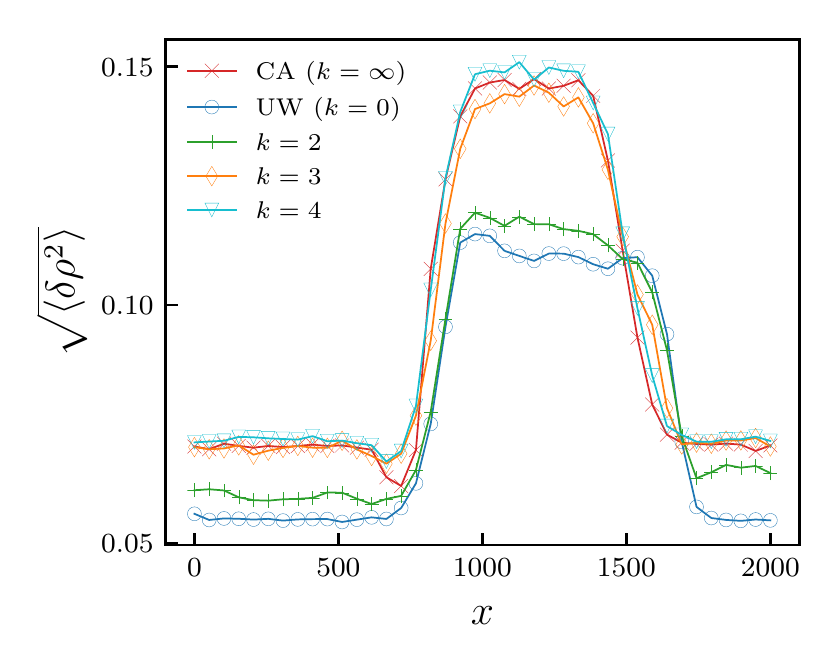}
		}
		\subfloat[]{\protect\protect\includegraphics[scale=0.6]{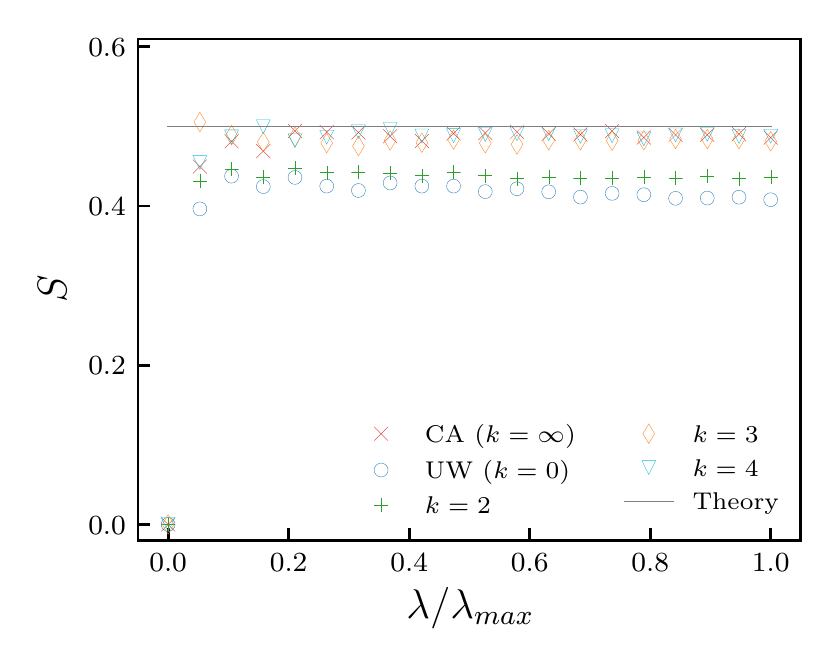}
		}
	\end{center}
	\protect\protect\caption{\label{fig:deterministic_flux} (a) Mean density and (b)
		standard deviation for a moving and diffusing  initial step function evolving according to \eqref{eq:FDDFTflux}, with $\delta\mathcal{E}/\delta \rho=\log \rho+u_0 x$ and $u_0=0.1$.
		Additinally, we report the structure factor of a uniform system in (c).
		IC: initial condition, CA: central approximation ($k=\infty$), UW: upwind approximation ($k=0$).
		Further values of $k$ are depicted to evaluate the spurious oscillations in the density and the artificial fluctuation dampening in case of both inhomogeneous (a-b) and homogeneous (c) systems.
		In what follows, we adopt a scheme with $k=3$, since it gives the compromise between accuracy in the sharp density profile and the fluctuations amplitude. The stochastic term is discretized according to Eq.~\eqref{eq:linear}. 
	}
\end{figure}

After selecting the adequate value of $k$, we proceed to the detailed
construction of the central and upwind deterministic fluxes in
\eqref{eq:centralupwind}:

\begin{enumerate}[label={\alph*)}]
	
	\item Upwind approximation of the deterministic flux: it is constructed as
	proposed in \cite{carrillo2015}, where a first- and second-order
	finite-volume method for nonlinear equations with gradient-flow
	structure is constructed. The equations treated in \cite{carrillo2015}
	have the form \eqref{eq:FDDFTflux} without the white noise $\noi$. The
	authors propose to firstly reconstruct the density profile in each cell
	$C_j$ as a constant profile for the first-order scheme, or as a linear
	profile for the second-order scheme,
	\begin{align}
	\widetilde{\rho}_j (x)=\begin{cases}
	\overline{\rho}_j,\quad x \in C_j, \quad \text{for the first-order scheme,}\\
	\overline{\rho}_j+\left(\rho_x\right)_j (x-x_j),\quad x \in C_j, \quad \text{for the second-order scheme,}
	\end{cases}
	\end{align}
	so that the east and the west density values $\rho_j^E$ and $\rho_j^W$ at the cell interfaces $x_{j+\frac{1}{2}}$ and $x_{j-\frac{1}{3}}$, respectively, are approximated as
	\begin{align}\label{eq:eastwest}
	\begin{gathered}
	\rho_j^E=\overline{\rho}_j+\frac{\Delta x}{2}\left(\rho_x\right)_j, \\
	\rho_j^W=\overline{\rho}_j-\frac{\Delta x}{2}\left(\rho_x\right)_j.
	\end{gathered}
	\end{align}
	The numerical derivatives $\left(\rho_x\right)_j$ at every cell $C_j$ are
	computed by means of an adaptive procedure which ensures that the point
	values \eqref{eq:eastwest} are second-order and non-negative. This
	procedure initially takes centred approximations of the form
	$\left(\rho_x\right)_j=\left(\overline{\rho}_{j+1}-\overline{\rho}_{j-1}\right)/(2\Delta
	x)$. If it then happens that $\rho_j^E<0$ or $\rho_j^W<0$, the scheme
	employs a \textit{minmod} limiter which ensures that the reconstructed
	values are non-negative as long as the cell averages $\overline{\rho}_j$ are
	non-negative,
	\begin{equation}
	\left(\rho_x\right)_j=\text{minmod}\left(\theta\,\frac{\overline{\rho}_{j+1}-\overline{\rho}_j}{\Delta x}, \frac{\overline{\rho}_{j+1}-\overline{\rho}_{j-1}}{2 \Delta x}, \theta\,
	\frac{\overline{\rho}_j-\overline{\rho}_{j-1}}{\Delta x}\right),
	\end{equation}
	where
	\begin{equation*}
	\text{minmod} \left(z_1,z_2,\ldots \right)= \begin{cases}
	\min\left(z_1,z_2,\ldots\right),\quad \text{if}\, z_i>0\quad \forall i,\\
	\max\left(z_1,z_2,\ldots\right),\quad \text{if}\, z_i<0\quad \forall i,\\
	0, \hspace{3.04cm}\text{otherwise.}
	\end{cases}
	\end{equation*}
	The parameter $\theta$ controls the numerical viscosity and it is taken to be $\theta=2$, as in Ref.~\cite{carrillo2015}.
	
	After completing the density reconstruction, the deterministic flux $F_{d,j+1/2}^u$ is evaluated with an upwind scheme as
	\begin{align}\label{eq:deterflux}
	F_{d,j+1/2}^u= u_{j+1/2}^+ \,\rho_{j}^E + u_{j+1/2}^{-} \,\rho_{j+1}^W,
	\end{align}
	where $u_{j+1/2}$ are discrete values computed from the central difference
	\begin{align}\label{eq:vel}
	u_{j+1/2} = -\frac{\left(\dfrac{\delta \mathcal{E}}{\delta \rho}\right)_{j+1}-\left(\dfrac{\delta \mathcal{E}}{\delta \rho}\right)_j}{\Delta x}.
	\end{align}
	The upwind formulation of the deterministic flux \eqref{eq:deterflux} is then accomplished by taking
	\begin{align}\label{eq:velpm}
	&u_{j+1/2}^+ = \max\left( u_{j+1/2},0 \right)\quad \text{and}\quad u_{j+1/2}^- = \min \left( u_{j+1/2},0 \right).
	\end{align}			
	Finally, the discrete variation of the free energy with respect to the density $\left(\dfrac{\delta \mathcal{E}}{\delta \rho}\right)_j$ is computed from \eqref{eq:varfreesimple}, in the case $g(s)=s$, as
	\begin{align}\label{eq:discretevel}
	\left(\dfrac{\delta \mathcal{E}}{\delta \rho}\right)_j= \Delta x \sum_i \kr(x_j-x_i) \rho_i + F(\rho_j) + V(x_j).
	\end{align}
	For general nonlinearities $g(s)$ a similar treatment is performed.
	
	\item Central approximation for the deterministic flux: this is the main
	strategy to treat the FH deterministic flux in the literature
	\cite{bell2007numerical,donev2010,kim2017}. In our case, given the
	generality of the free energy in \eqref{eq:freesimple}, we propose to
	evaluate the central deterministic flux as
	\begin{align}\label{eq:deterfluxc}
	F_{d,j+1/2}^c= u_{j+1/2} \,\rho_{j+1/2},
	\end{align}
	where $u_{j+1/2}$ is computed as in \eqref{eq:vel}, with the discrete variation of the free energy satisfying \eqref{eq:discretevel}, and $\rho_{j+1/2}$ is taken as the averaged from the adjacent cells,
	\begin{align}\label{eq:detdensity}
	\rho_{j+1/2}= \frac{\overline{\rho}_{j}+\overline{\rho}_{j+1}}{2}.
	\end{align}
	Classical hybrid schemes employ a high-order approximation for the central approximation of the deterministic flux. For this work, however, we just consider the low-order differences \eqref{eq:vel} and \eqref{eq:detdensity}, given that the presence of the stochastic flux limits the spatial order of accuracy.
	Previous works in the literature also propose this low-order central differences \cite{bell2007numerical,donev2010,kim2017}.
\end{enumerate}

\subsection{Stochastic flux}

The evaluation of the stochastic flux \eqref{eq:fluxes} must be done
carefully since the divergence of the white noise $\noi$ cannot be evaluated
pointwise in time and space. This problem is typically overcome by
	evaluating the noise in the cell by means of a spatiotemporal average,
	following  \cite{donev2010} and subsequently employed by Donev and
	collaborators in
	\cite{balakrishnan2014fluctuating,bhattacharjee2015fluctuating,kim2017} ,

\begin{equation}\label{eq:stoflux0}
\overline{\noi}_j=\frac{1}{\Delta x \Delta t}\int_t^{t+\Delta t} \int_{x_{j-\frac{1}{2}}}^{x_{j+\frac{1}{2}}} \noi(x,t) dx\,dt,
\end{equation}
which, by the definition of the white noise, is equal to a normal distribution with zero mean and variance $\left(\Delta x \Delta t\right)^{-1}$, so that
\begin{equation}\label{eq:gaussian}
\overline{\noi}_j=\mathcal{N}(0,1)/ \sqrt{\Delta x \Delta t}.
\end{equation}
Several approximations for the stochastic flux have been put forward in the literature \cite{bell2007numerical,kim2017}.
	They rely on computing the stochastic flux directly at the interfaces using a random number generator, and we refer the reader to \cite{donev2010} for more details about this approach.
	In this work, however, we aim to employ the spatiotemporal cell average in
	Eq.~\eqref{eq:stoflux0} to compute the stochastic fluxes at the interfaces.
	We are inspired by the literature on numerical methods for hyperbolic
	problems where it is common to evaluate fluxes in a central or upwind
	fashion. Of course, here we are not aiming to achieve a higher accuracy at
	the interface, given that the cell averages are sampled from a distribution.
We test the following four different approximations for the stochastic flux, which are
compared in Sect.~\ref{numerical_applications_eq}:
\begin{enumerate}[label=(\alph*)]
	\item Forward approximation of the form
	\begin{equation}\label{eq:forward}
	F_{s,j+1/2}=\left(\sqrt{\frac{\rho}{\beta}}\noi \right)_{j+1/2}=\sqrt{\frac{\rho_{j}}{\beta}} \overline{\noi}_{j}.
	\end{equation}
	\item Linear approximation of the form
	\begin{equation}\label{eq:linear}
	F_{s,j+1/2}=\left(\sqrt{\frac{\rho}{\beta}}\noi \right)_{j+1/2}=\sqrt{\frac{\rho_{j+1/2}}{\beta}} \noi_{j+1/2},
	\end{equation}
	where
	\begin{equation}\label{eq:linear2}
	\rho_{j+1/2}=\frac{\overline{\rho}_{j}+\overline{\rho}_{j+1}}{2},\quad \noi_{j+1/2}=\frac{\overline{\noi}_{j}+\overline{\noi}_{j+1}}{2}.
	\end{equation}
	\item Parabolic approximation of the form
	\begin{equation}\label{eq:parabolic}
	F_{s,j+1/2}=\left(\sqrt{\frac{\rho}{\beta}}\noi \right)_{j+1/2}=\sqrt{\frac{\rho_{j+1/2}}{\beta}} \noi_{j+1/2},
	\end{equation}
	where
	\begin{equation}\label{eq:parabolic2}
	\begin{gathered}
	\rho_{j+1/2}=\alpha_1\left(\overline{\rho}_{j-1}+\overline{\rho}_{j+2}\right)+\alpha_2\left(\overline{\rho}_{j}+\overline{\rho}_{j+1}\right),\\
	\noi_{j+1/2}=\alpha_1\left(\overline{\noi}_{j-1}+\overline{\noi}_{j+2}\right)+\alpha_2\left(\overline{\noi}_{j}+\overline{\noi}_{j+1}\right),\\
	\alpha_1=(1-\sqrt{3})/4, \quad \alpha_2=(1+\sqrt{3})/4.\\
	\end{gathered}
	\end{equation}
	The coefficients $\alpha_1$ and $\alpha_2$ are selected as in \cite{bell2007numerical}, with the objective of preserving both the average and the variance in each time step.
	\item Upwind approximation, where $\noi_j$ is taken as the stochastic velocity, so that a similar expression to the deterministic flux in \eqref{eq:deterflux} is taken,
	\begin{equation}\label{eq:stoflux}
	F_{s,j+1/2}=\left(\sqrt{\frac{\rho}{\beta}}\noi \right)_{j+1/2}=\sqrt{\frac{\rho_{j}^E
		}{\beta}} \noi_{j+1/2}^++\sqrt{\frac{\rho_{j+1}^W
		}{\beta}} \noi_{j+1/2}^-,
	\end{equation}
	where
	\begin{align}\label{eq:noipm}
	&\noi_{j+1/2}^+ = \max\left( \noi_{j+1/2},0 \right),\quad  \noi_{j+1/2}^- = \min \left( \noi_{j+1/2},0 \right),
	\end{align}	
	and $\noi_{j+1/2}=(\overline{\noi}_{j}+\overline{\noi}_{j+1})/2$. The east and west density values $\rho_{j}^E$ and $\rho_{j}^W$ are computed as in the deterministic flux, either with a first- or second-order reconstruction \eqref{eq:eastwest}.
\end{enumerate}

\subsection{Stochastic time integrators}
\label{time_integrators}

The derivation of the temporal integrators to advance in time the semidiscrete equation \eqref{eq:fvfluxes} is accomplished by the equation
\begin{align}
d\boldsymbol {\overline{\rho}}(t)= {\boldsymbol {\mu }} (\boldsymbol {\overline{\rho}}(t) )\,dt+ \boldsymbol {\sigma} (\boldsymbol {\overline{\rho}}(t)) \,\boldsymbol{\overline{\noi}} \ dt,
\label{eq:discretized_system}
\end{align}
where the vectors $\boldsymbol {\overline{\rho}}(t)$ and $\boldsymbol {\overline{\noi}}$ contain the cell averages defined in \eqref{eq:rhoaver} and \eqref{eq:stoflux0}, respectively, so that $\boldsymbol {\overline{\rho}}(t)=(\overline{\rho}_{1}(t),\overline{\rho}_{2}(t),\ldots ,\overline{\rho}_{n}(t))$ and $\boldsymbol {\overline{\noi}}(t)=(\overline{\noi}_{1}(t),\overline{\noi}_{2}(t),\ldots ,\overline{\noi}_{n}(t))$. The vector ${\boldsymbol {\mu }} (\boldsymbol {\overline{\rho}}(t) )$ and the matrix $\boldsymbol {\sigma} (\boldsymbol {\overline{\rho}}(t))$ depend on the density cell averages $\boldsymbol {\overline{\rho}}(t)$ and their structures vary depending on the choice of the deterministic and stochastic fluxes, respectively.

From Eq.~\eqref{eq:discretized_system} we employ It\^o's lemma to
approximate the two functions ${\boldsymbol {\mu }} (\boldsymbol
{\overline{\rho}}(t) )$ and $\boldsymbol {\sigma} (\boldsymbol
{\overline{\rho}}(t))$. After integrating in time then we obtain the Taylor
expansion of the stochastic process. Truncating this expansion with an error
$\mathcal{O}(\Delta t^{1/2})$ and integrating between $t$ and $t+\Delta t$,
one can derive the following family of implicit-explicit Euler-Maruyama
integrators~\cite{Kloeden1992}, whose component-wise form satisfies
\begin{align}
\overline{\rho}_j(t+\Delta t)= \overline{\rho}_j(t) +  \left[ \left( 1-\theta \right) \mu_j (\boldsymbol {\overline{\rho}}(t)) + \theta\,\mu_j (\boldsymbol {\overline{\rho}}(t+\Delta t))  \right] \Delta t + \sum_{k=1}^n \sigma_{jk}(\boldsymbol {\overline{\rho}}(t))   \overline{W}_k (t) \Delta t.
\label{eq:euler-maruyama}
\end{align}
The parameter $\theta$ allows us to have an explicit ($\theta=0$), implicit ($\theta=1$) or semi-implicit ($\theta=0.5$) temporal integrator. Euler-Maruyama is the highest order integrator for which no multiple stochastic integrals have to be computed, but it has only $0.5$ strong order of convergence.

Keeping in the expansion all the terms up to $\mathcal{O}( \Delta t)$, one
obtains a derivative-free family of implicit-explicit Milstein integrators
with strong order $1.0$ and weak order $0.5$~\cite{Kloeden1992}. The
component-wise version of this scheme is
\begin{equation}
\begin{split}
\overline{\rho}_j(t+\Delta t)= &\overline{\rho}_j(t) +  \left[ \left( 1-\theta \right) \mu_j (\boldsymbol {\overline{\rho}}(t)) + \theta\,\mu_j (\boldsymbol {\overline{\rho}}(t+\Delta t))  \right] \Delta t + \sum_{k=1}^n \sigma_{jk}(\boldsymbol {\overline{\rho}}(t))  \overline{W}_k (t) \Delta t  \\
& + \frac{1}{\sqrt{ \Delta t} } \sum_{l,m=1}^{n}   \left[  \sigma_{jm} (  \boldsymbol {\Upsilon}_{l}(t) ) - \sigma_{jl}(\boldsymbol {\overline{\rho}}(t) ) \right] I_{l,m}(t),
\label{eq:milstein}
\end{split}
\end{equation}
where the $l$-th row of the matrix $\boldsymbol {\Upsilon}$ is defined as
\begin{align}
\boldsymbol {\Upsilon}_{l}(t) = \boldsymbol {\overline{\rho}}(t) + {\boldsymbol {\mu} (\boldsymbol {\overline{\rho}}(t) ) }  \Delta t + {\boldsymbol {\sigma }}_{l} (\boldsymbol {\overline{\rho}}(t) )  \sqrt{ \Delta t},
\label{eq:milstein_supporting}
\end{align}
and multiple stochastic integrals $I_{l,m}(t)= \int_t^{t+ \Delta t} \noi^l \ \noi^m dt$, where $\noi_l$ and $\noi_m$ are two white noises.
These integrals do not have a simple analytical solutions, thus are approximated as function of the white noise cells average in Eq.~\eqref{eq:stoflux0} as \cite{lutsko2015two}:
\begin{align}
I_{l,m}(t)=
\begin{cases}
\frac{1}{2} \left[ \left( \overline{W}^l \right) ^2 - 1 \right] \Delta t  \quad &\text{if} \quad  l=m ,\\[7pt]
\frac{\Delta t}{2} \overline{W}^l \overline{W}^m   + \sqrt{k_p  \Delta t } ( \varphi_{l} \overline{W}^m - \varphi_{m} \overline{W}^l ) &\\[5pt]
\hspace{2cm}+  \sum_{r=1}^p \frac{1}{2 \pi r} \left[ \zeta_{lr} ( \sqrt{2} \ \overline{W}^m \sqrt{\Delta t }  + \eta_{m} )  -\zeta_{mr} ( \sqrt{2} \ \overline{W}^l \sqrt{\Delta t } + \eta_{l} ) \right] \quad &\text{otherwise},
\end{cases}
\label{eq:multiple_integrals}
\end{align}
where $\varphi_{l}$, $\zeta_{lr}$ and $\eta_{m}$ are pairwise independent variables with distribution $\mathcal{N}(0,\Delta t)$ and $k_p$ is given by
\begin{align}
k_p= \frac{1}{12} - \frac{1}{2 \pi^2} \sum_{1}^{p} \frac{1}{r^2}.
\end{align}
The value $p$ determines the accuracy of the multiple stochastic integral approximation and subsequently the accuracy of the scheme.
A value of $p= k/ \Delta t $ for some constant $k$ is enough to preserve the accuracy of the scheme \cite{Kloeden1992}.

Stochastic time integration schemes of higher strong order have also been proposed in the literature \cite{Kloeden1992}.
However, these schemes are very computationally expensive due to the presence
of high-order multiple stochastic integrals to be solved.
Moreover, in many physical applications, the convergence in probability, also called weak convergence, is more relevant than the strong convergence.
For this reason, a last time integration scheme we will study the following explicit weak order $2.0$ Runge-Kutta scheme:
\begin{align}
\begin{split}
& \widetilde{\rho}_j(t+\Delta t)= \overline{\rho}_j(t) + \frac{1}{2} \left[  \mu_j (\boldsymbol {\Upsilon}(t) ) + \mu_j (\boldsymbol {\overline{\rho}}(t))  \right] \Delta t + \Phi(t) , \\
& \overline{\rho}_j(t+\Delta t)=  \overline{\rho}_j(t) + \frac{1}{2} \left[  \mu_j ( \widetilde{\boldsymbol {\rho}} (t+\Delta t) ) + \mu_j (\boldsymbol {\overline{\rho}}(t))  \right] \Delta t + \Phi(t) ,
\label{eq:RK2}
\end{split}
\end{align}
where the vector $\Phi(t)$ has components:
\begin{align}
\begin{split}
\Phi_j(t) =& \frac{1}{4} \sum_{l=1}^n  \left[  \boldsymbol {\sigma}_{lj}( \boldsymbol {\Lambda}_{l +}(t) ) + \boldsymbol {\sigma}_{lj}( \boldsymbol {\Lambda}_{l -}(t) )  + 2 \boldsymbol {\sigma}_{lj}(\boldsymbol {\overline{\rho}}(t))  \right]  \overline{W}^l (t) \sqrt{\Delta t } \\&+\frac{1}{4} \sum_{l=1}^n \sum_{r=1, r \neq l }^n \left[  \boldsymbol {\sigma}_{lj}( \boldsymbol {\Xi}_{r +}(t) ) + \boldsymbol {\sigma}_{lj}( \boldsymbol {\Lambda}_{r -}(t) )  - 2 \boldsymbol {\sigma}_{lj}(\boldsymbol {\overline{\rho}}(t))  \right] \overline{W}^l (t)  \\
&+\frac{1}{4} \sum_{l=1}^n  \left[  \boldsymbol {\sigma}_{lj}( \boldsymbol {\Lambda}_{l +}(t) ) - \boldsymbol {\sigma}_{lj}( \boldsymbol {\Lambda}_{l -}(t) ) \right] \left[ \left( \overline{W}^l (t) \right)^2  - 1 \right]  \sqrt{\Delta t} \\
&+\frac{1}{4 } \sum_{l=1}^n \sum_{r=1, r \neq l }^n \left[  \boldsymbol {\sigma}_{lj}( \boldsymbol {\Xi}_{r +}(t) ) - \boldsymbol {\sigma}_{lj}( \boldsymbol {\Xi}_{r -}(t) ) \right] \left[ \overline{W}^l (t) \overline{W}^r (t) + V_{r,j} \right] \sqrt{\Delta t} , \\
\label{eq:RK2_Phi}
\end{split}
\end{align}
and the supporting values:
\begin{align}
& \boldsymbol {\Upsilon} = \boldsymbol {\overline{\rho}}(t) + {\boldsymbol {\mu }}(\boldsymbol {\overline{\rho}}(t)) \Delta t + \sum_{j=1}^n\boldsymbol {\sigma}_{j}(\boldsymbol {\rho}(t))  \Delta W_j (t) ,\\
& \boldsymbol {\Lambda}_{l \pm} = \boldsymbol {\overline{\rho}}(t)  + {\boldsymbol {\mu }}(\boldsymbol {\overline{\rho}}(t))  \Delta t \pm \boldsymbol {\sigma}_{l}(\boldsymbol {\overline{\rho}}(t))  \sqrt{\Delta t } , \\
& \boldsymbol {\Xi}_{l \pm} = \boldsymbol {\overline{\rho}}(t)  \pm \boldsymbol {\sigma}_{l}(\boldsymbol {\overline{\rho}}(t))  \sqrt{\Delta t } .
\label{eq:rk2_supporting}
\end{align}
The random matrix $\mathbf{V}$ is defined as:
\begin{align}
V_{r,j}(t)=
\begin{cases}
& \pm 1 \quad \text{with} \quad p=\frac{1}{2}  \quad \text{if} \quad r<j ,\\
&  - 1 \quad \text{if} \quad r=j ,\\
& - V_{j,r}(t) \quad \text{if} \quad r>j ,
\end{cases}
\label{eq:V_matrix}
\end{align}
where $p$ indicates the probability. It has to be emphasised that such a
scheme does not involve the computation of multiple stochastic integrals,
thus its strong order of convergence is expected to be at most $1.0$.

\subsubsection{Weak and strong order of convergence for temporal integrators}\label{subsec:strongweak}

The order of convergence can be measured in the strong and weak sense, for
which the strong and weak errors are respectively defined for a particular
time $\tau$ and a group of trajectories $\Gamma=\left\{\gamma_1,
\gamma_2,\ldots,\gamma_m\right\}$ as
\begin{equation}
\epsilon_s=\left< \left| \boldsymbol{\overline{\rho}}^{\gamma}(\tau) -  \boldsymbol{\overline{\rho}}_{exact}^{\gamma}(\tau) \right| \right>_{\gamma\,\in\,\Gamma} \quad\text{and}\quad \epsilon_w= \left|\left<  \boldsymbol{\overline{\rho}}^{\gamma}(\tau) \right>_{\gamma\,\in\,\Gamma} - \left< \boldsymbol{\overline{\rho}}^{\gamma}_{exact}(\tau)\right>_{\gamma\,\in\,\Gamma}\right|,
\end{equation}
where $\boldsymbol{\overline{\rho}}^{\gamma}(\tau)$ refers to the numerical density cell averages at time $\tau$ following trajectory $\gamma$, $\boldsymbol{\overline{\rho}}_{exact}^{\gamma}(\tau)$ denotes the exact or reference solution which is considered to be the true solution of the stochastic equation, the ensemble average $\left< \cdot \right>$ is taken over the trajectories ${\gamma\,\in\,\Gamma}$, and the norm $\left| \cdot \right|$ is taken to be the standard $L^1$-norm.

In Fig.~\ref{fig:cpu_time} we evaluate the strong and weak errors for the
described stochastic integrators. They are obtained by simulating equation
\eqref{eq:discretized_system} in the simplified case of geometric Brownian
motion, for which ${\boldsymbol {\mu }} (\boldsymbol {\overline{\rho}}(t)
)=-\boldsymbol {\overline{\rho}}(t)$ and $\boldsymbol {\sigma} (\boldsymbol
{\overline{\rho}}(t))=0.5 \boldsymbol {\overline{\rho}}(t)$, thus eliminating
the spatial derivatives. As a result, the temporal evolution of the density
for a cell $j$, which is independent from the rest of cells, follows
\begin{equation}\label{eq:geobrownian}
d \overline{\rho}_j (t)= -\overline{\rho}_j dt +0.5 \overline{\rho}_j \overline{\noi}_j dt,
\end{equation}
with the cell averaged white noise $\overline{\noi}_j$ defined as in \eqref{eq:stoflux0}. For the simulation we selected $\overline{\rho}_j (0) =1$. Geometric Brownian motion is useful to compute the strong and weak errors since the exact solution in analytically known \cite{oksendal2003stochastic}.

The results in Fig.~\ref{fig:cpu_time} (a) and (b) depict the strong and weak
order of convergence for the temporal integrators. Concerning the former, as
expected the Euler-Maruyama presents an order of $0.5$, while Milstein an
order of $1.0$. Runge-Kutta is expected to have a strong order of at least
$0.5$, and in the plot it approaches a value of $1.0$.

With respect to the weak order, the whole families of Euler-Maruyama and
Milstein solvers are expected to have an order of $1.0$, while the
Runge-Kutta an order of $2.0$. Such theoretical predictions are respected for
all schemes, with the exception of the semi-implicit methods which
outperforms, giving an order between $1.0$ and $2.0$.

On Fig.~\ref{fig:cpu_time} (c) we plot the cpu time against the total number of cells $n$ for each of the temporal integrators.
The Euler-Maruyama accounts for $\mathcal{O}(n)$ computations, the Milstein for $\mathcal{O}(n^2)$, and the Runge-Kutta for $\mathcal{O}(n^3)$.
However, for $n< 100$ we get a lower cpu time for Runge-Kutta, if compared with all the other integrators except for the explicit Euler-Maruyama.

\begin{figure}[ht!]
	\begin{center}
		\subfloat[]{ \protect\protect\includegraphics[scale=0.8]{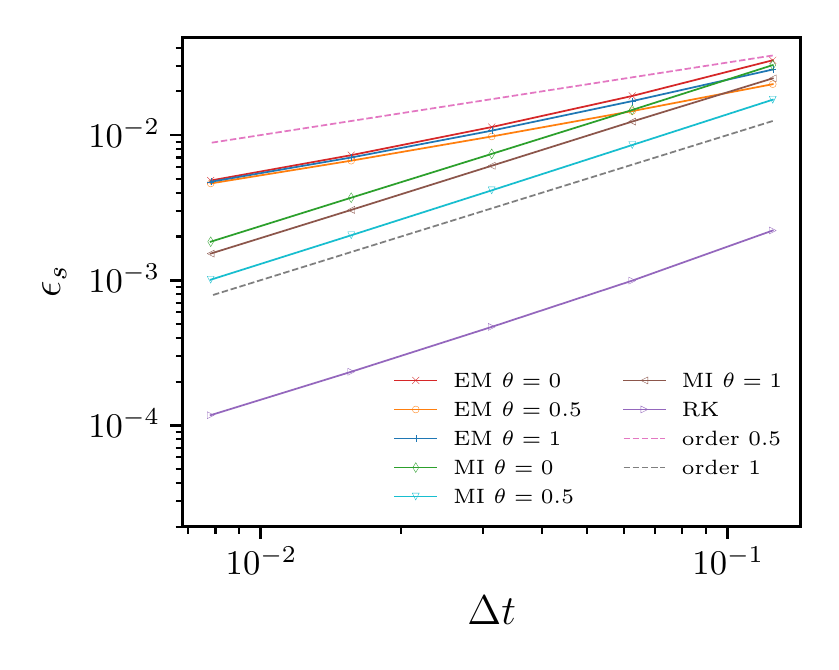}
		}
		\subfloat[]{ \protect\protect\includegraphics[scale=0.8]{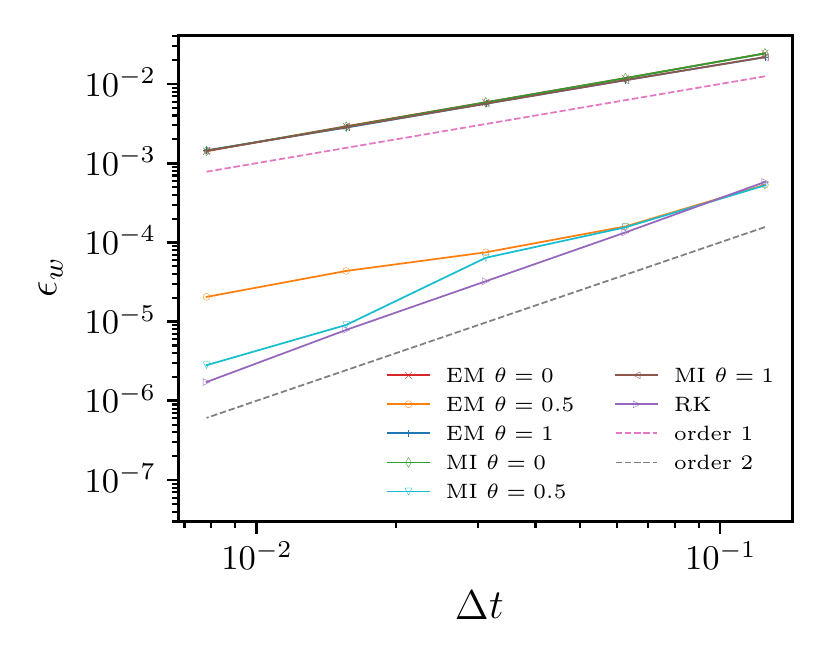}
		}
		
		\subfloat[]{
			\protect\protect\includegraphics[scale=0.8]{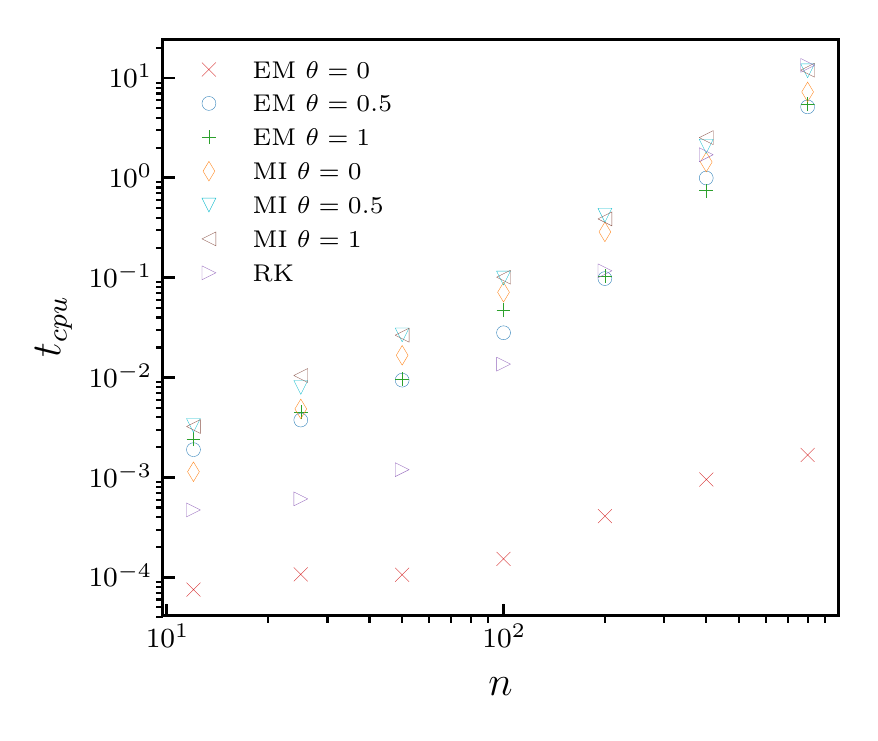}
		}
	\end{center}
	\caption{ \label{fig:cpu_time} Strong (a) and weak (b) errors convergence for geometric Brownian motion.
		In (c) we report the cpu time for each time integration schemes as a function
		of the number of cells $n$. EM: Euler-Maruyama, MI: Milstein, RK:
		Runge-Kutta. Explicit ($\theta=0$), semi-implicit ($\theta=0.5$) and implicit
		($\theta=1$). }
\end{figure}

\subsection{Positivity of the density through an adaptive time step}
\label{adaprive_time_step}

A natural constraint for physical systems is the positivity of the density
field, and the numerical solution is expected to satisfy such a requirement.
Numerical schemes with the property of preserving the positivity of the
density have been developed in the literature, specially in the context of
deterministic conservation law
PDEs~\cite{carrillo2015,bessemoulin2012finite}. The strategy is usually to
derive a Courant-Friedrichs-Lewy (CFL) condition which imposes a constraint
for the maximum $\Delta t$ so that the density always remains non-negative.

For particular discretizations of the numerical fluxes \eqref{eq:fluxes} it
is also possible to derive a CFL condition for the SPDE
in~\eqref{eq:FDDFToverd}. This CFL condition depends on the Gaussian
distributions from the white noise \eqref{eq:gaussian}, as well as on the
density profile. Following the derivation provided in \cite{carrillo2015} for
a deterministic gradient-flow equation, we proceed to provide an example of
the CFL derivation when the upwind discretizations \eqref{eq:deterflux} and
\eqref{eq:stoflux} for the deterministic and stochastic fluxes, respectively,
are employed.

\begin{lemma}
	Consider the SPDE \eqref{eq:FDDFToverd} with initial data $\rho_0(x)>0$,
	together with the semi-discrete finite-volume scheme \eqref{eq:fvfluxes} with
	the upwind discretizations for the deterministic
	\eqref{eq:deterflux}-\eqref{eq:discretevel} and stochastic
	\eqref{eq:gaussian}, \eqref{eq:stoflux} fluxes. Assume that the SPDE  is
	temporally discretized with a deterministic Euler forward method. Then, the
	computed cell averages satisfy $\overline{\rho}_j\geq 0$, $\forall j$,
	provided that the following two CFL conditions for $\Delta t$ hold:
	\begin{equation}\label{eq:positcond}
	\frac{1}{2}-\lambda_1 u^+_{j+\frac{1}{2}}- \lambda_2 \gau^+_{j+\frac{1}{2}}/\sqrt{\rho_j^E \beta}\geq 0,\quad \frac{1}{2}-\lambda_1 u^-_{j-\frac{1}{2}}-\lambda_2 \gau^-_{j-\frac{1}{2}}/\sqrt{\rho_j^W \beta}\geq 0,
	\end{equation}
	where \begin{equation}\label{eq:positcond2}
	\lambda_1\coloneqq\frac{\Delta t}{\Delta x}, \quad \lambda_2\coloneqq\sqrt{\frac{\Delta t}{\Delta x}}, \quad  \overline{\gau}_j=\overline{\mathcal{W}}_j \sqrt{\Delta x \Delta t}=\mathcal{N}(0,1),
	\end{equation}
	and $\gau_{j+1/2}^+$, $\gau_{j+1/2}^-$ are constructed as in \eqref{eq:noipm}, so that
	\begin{align}\label{eq:gaupm}
	&\gau_{j+1/2}^+ = \max\left( \gau_{j+1/2},0 \right),\quad  \gau_{j+1/2}^- = \min \left( \gau_{j+1/2},0 \right).
	\end{align}	
\end{lemma}
\begin{proof} Assume that for a given time $t$ the computed solution for the density is known and positive: $\rho_j(t)\geq0$, $\forall j$.
	The new cell averages following a forward Euler temporal scheme in the
	finite-volume formulation \eqref{eq:fvfluxes} satisfy
	\begin{equation}
	\overline{\rho}_j(t+\Delta t)=\overline{\rho}_j(t)-\Delta t \left[\frac{F_{d,j+1/2}-F_{d,j-1/2}}{\Delta x}+\frac{F_{s,j+1/2}-F_{s,j-1/2}}{\Delta x}\right].
	\end{equation}
	Substituting the deterministic and stochastic fluxes for their
	upwinded discretizations \eqref{eq:deterflux} and \eqref{eq:stoflux},
	respectively, and by employing the notation specified in
	\eqref{eq:positcond2}, it follows
	\begin{equation}\label{eq:positcond3}
	\begin{split}
	\overline{\rho}_j (t+\Delta t) = &\frac{1}{2} (\rho_j^E+\rho_j^W)-\lambda_1 \left[u^+_{j+\frac{1}{2}}\rho_j^E+u^-_{j+\frac{1}{2}}\rho_{j+1}^W-u^+_{j-\frac{1}{2}}\rho_{j-1}^E-u^-_{j-\frac{1}{2}}\rho_{j}^W\right]\\
	&-\lambda_2 \left[\gau^+_{j+\frac{1}{2}}\sqrt{\rho_j^E/\beta}+\gau^-_{j+\frac{1}{2}}\sqrt{\rho_{j+1}^W/\beta}-\gau^+_{j-\frac{1}{2}}\sqrt{\rho_{j-1}^E/\beta}-\gau^-_{j-\frac{1}{2}}\sqrt{\rho_{j}^W/\beta}\right]\\
	= &\lambda_1 \left[-u^-_{j+\frac{1}{2}}\rho_{j+1}^W + u^+_{j-\frac{1}{2}}\rho_{j-1}^E\right ]+\lambda_2 \left[-\gau^-_{j+\frac{1}{2}}\sqrt{\rho_{j+1}^W/\beta}+\gau^+_{j-\frac{1}{2}}\sqrt{\rho_{j-1}^E/\beta} \right ]\\
	& +\left[\frac{1}{2}-\lambda_1 u^+_{j+\frac{1}{2}}- \lambda_2 \gau^+_{j+\frac{1}{2}}/\sqrt{\rho_j^E \beta}\right]\rho_j^E+\left[\frac{1}{2}-\lambda_1 u^-_{j-\frac{1}{2}}-\lambda_2 \gau^-_{j-\frac{1}{2}}/\sqrt{\rho_j^W \beta}\right]\rho_j^W.
	\end{split}
	\end{equation}
	Due to the fact that the reconstructed point values for the density
	$\rho_{j-1}^E$, $\rho_{j+1}^E$, $\rho_{j}^W$ and $\rho_{j+1}^W$ are
	non-negative,  and bearing in mind that $u^+_{j-\frac{1}{2}}$,
	$\gau^+_{j-\frac{1}{2}}\leq 0$ and $u^-_{j+\frac{1}{2}}$,
	$\gau^-_{j+\frac{1}{2}}\geq 0$ due to \eqref{eq:velpm} and \eqref{eq:gaupm},
	it follows $\overline{\rho}_j (t+\Delta t)\geq 0$, $\forall j$, provided
	that the CFL conditions \eqref{eq:positcond} hold.
\end{proof}

The CFL conditions in \eqref{eq:positcond} ensure that the density remains
non-negative at all times, independently of the values produced by the normal
distributions of the white noise spatio-temporal average \eqref{eq:gaussian}.
In the case of a rare event in which the Gaussian distribution produces
low-probability values located at the tails of the distribution, $\Delta t$
would be adapted accordingly to ensure the positivity.
However, this adaptive time step strategy entails two main disadvantages.
First, it requires at each time step the solution of a second-order equation (in
1D) or a two-parameter equation in multi-dimensional problems.
Second, since the time-step size is dependent on the random number at each
step, higher (or lower) $\Delta t$ may be favored by some random numbers,
thus not guaranteeing that the correct Brownian path is
followed \cite{gaines1997variable}.

Previous works in the literature have already addressed the issue of positivity by means of varied approaches.
In the context of FH, the authors of Ref.~\cite{kim2017} have effectively
opted for introducing cutting functions based on smoothed Heavisides 
which prevent the density from becoming negative.
The main drawbacks of this strategy are 1) despite reducing the chances
of having negative density values, positive densities are not guaranteed, and
2) it affects the density distribution.

A further alternative to preserve positivity lays in the concept of Brownian
trees, which were firstly introduced in \cite{gaines1997variable} in order to
address the numerical resolution of stochastic differential equations with
variable time steps. The key idea here is that it is vital to respect the Brownian path that is
formed after evaluating the normal distributions \eqref{eq:gaussian}.
This means that upon advancing our simulation from time $t$ a certain $\Delta
t_1$ and realising that the density in one of the nodes $j$ has become
negative we cannot just simply repeat the time step with a shorter $\Delta
t_2<\Delta t_1$ in order to maintain positivity. The values of the normal
distributions after the first trial of advancing $\Delta t_1$ have to be
respected if the Brownian path is to be preserved. In addition, those values
of the normal distributions at $t+\Delta t_1$ have to be employed when
computing the values at $t+\Delta t_2$, even if the jump from $t$ to
$t+\Delta t_1$ has produced negative densities.

The solution to effectively take the statistical information at $t +\Delta t$
into account when repeating the time step is the so-called Brownian bridge
\cite{sotiropoulos2008adaptive,lutsko2015two}. It allows the computation of
$\overline{\noi}_j$ in Eq.~\eqref{eq:stoflux0} at an intermediate time step
$t +\Delta t/2$ by means of the formula
\begin{align}\label{eq:brownian}
\overline{\noi}_j\left(t+ \frac{\Delta t}{2}\right)-\overline{\noi}_j(t) = \frac{\overline{\noi}_j\left(t+ \Delta t\right)-\overline{\noi}_j\left(t\right)}{2} + \mathcal{N}\left(0, \frac{\Delta t}{4} \right).
\end{align}
As a result, our tactic consists in initially selecting an adequately small
$\Delta t$. Then, if after some time the density becomes negative, $\Delta t$
is halved to compute the intermediate time step from the Brownian
bridge \eqref{eq:brownian}. If that intermediate state leads to further
negative densities, the Brownian bridge is applied as many times as needed.
The information at $t +\Delta t$ is saved to be employed once all the
intermediate time steps with non-negative densities are computed. A pseudocode
to implement the Brownian bridge is written in Algorithm \ref{alg.mainLoop}.
As a remark, the adequate choice of a small initial $\Delta t$ for the
simulation is essential to reduce the number of Brownian bridges to a
minimum. A compromise is of course needed, since an extremely small $\Delta
t$ does not lead to negative densities but requires a high computational cost
for the simulation.

\begin{algorithm}[H]
	\tiny
	\KwIn{$\rho(t)$}
	\KwOut{$\rho(t+\Delta t)$}
	NegativeDensity=True\;
	$\Delta t= \Delta t_0$\;
	$partitions=0$\;
	
	\While{$(\text{NegativeDensity==True})$}{
		NegativeDensity=False\;
		$\rho_{tmp}=\rho$\;
		
		\For{$i \gets 0$ \textbf{to} $2^{partitions}$} {
			compute {\bf Brownian bridge}\;
			update $\rho_{tmp}$\;
			
			\If{$(\text{any}(\rho_{tmp})<0)$}{
				NegativeDensity=True\;
			}
		}	
		$\Delta t \gets \Delta t/2$ \;
		$partitions \gets partitions+1$\;
	}
	$\rho(t+\Delta t)  \gets \rho_{tmp}$\;
	\Return{$\rho(t+\Delta t)$}\;
	\caption{Algorithm adopted to overcome the issue of negative density. It is based on an adaptive timestep combined with the Brownian bridge technique,
		that allows to preserve the properties of the probability distribution underlying the stochastic process.}
	\label{alg.mainLoop}
\end{algorithm}

\subsection{Boundary conditions}\label{subsec:boundary}
In this section we analyse the implementation of boundary conditions for the
cases of periodic, confined and open systems. For systems with a periodic
boundary, it is sufficient to impose
\begin{align}
\rho_0= \rho_N.
\end{align}

For no-flux conditions, the boundary conditions to impose on the fluxes are
\begin{align}
F_{j\pm 1/2} = 0 \quad \text{for}\quad j=0,\,N.
\end{align}

Open systems in thermal and chemical equilibrium with a reservoir can be
represented by a $\mu VT$ ensemble with constant grand potential $\Omega
[\rho]= \mathcal{E}[\rho] - \mu \int \rho dx$, where $\mathcal{E}[\rho]=\mathcal{F}[\rho] + \int V(x)\, \rho dx$ with $\mathcal{F}[\rho]$ being the Helmholtz free-energy functional, $V(x)$ the external potential acting on the system and $\mu$
the chemical potential. Using the fact that the functional derivative of
$\Omega$ with respect to $\rho$ is null in equilibrium, we obtain $\delta
\mathcal{E}[\rho]/\delta \rho=
\mu$. Since the system is assumed to be in contact with a reservoir at
temperature $T_{\text{res}}$ and chemical potential $\mu_{\text{res}}$, the corresponding boundary
condition applied to compute the velocities $u_{j+1/2}$ at the boundaries in  \eqref{eq:vel} is
\begin{align}\label{eq:varfreebound}
\left(\frac{\delta \mathcal{E}}{\delta \rho}\right)_0=\left(\frac{\delta \mathcal{E}}{\delta \rho}\right)_N= \mu_{\text{res}},
\end{align}
with $\mu_{\text{res}}$ being the chemical potential of the reservoir. From the value of $\mu_{\text{res}}$ one can compute the density by solving \eqref{eq:varfreesimple} for a fixed value of $\delta \mathcal{E}/\delta \rho$. This implies that the values of $\rho_0$, $\rho_N$ and any additional ghost node are imposed from \eqref{eq:varfreebound} for all times. Depending on the particular choice of free energy in \eqref{eq:freesimple}, it may be possible to converge to different density profiles depending on the initial condition for the iterative algorithm to solve \eqref{eq:varfreesimple}. This open boundary condition imposes a positive or negative flux of mass through the boundary, and as a result the total mass is not conserved in time.

%
%
%
%

\begin{figure}[t!]
	\begin{center}
		\protect\protect\includegraphics[scale=0.48]{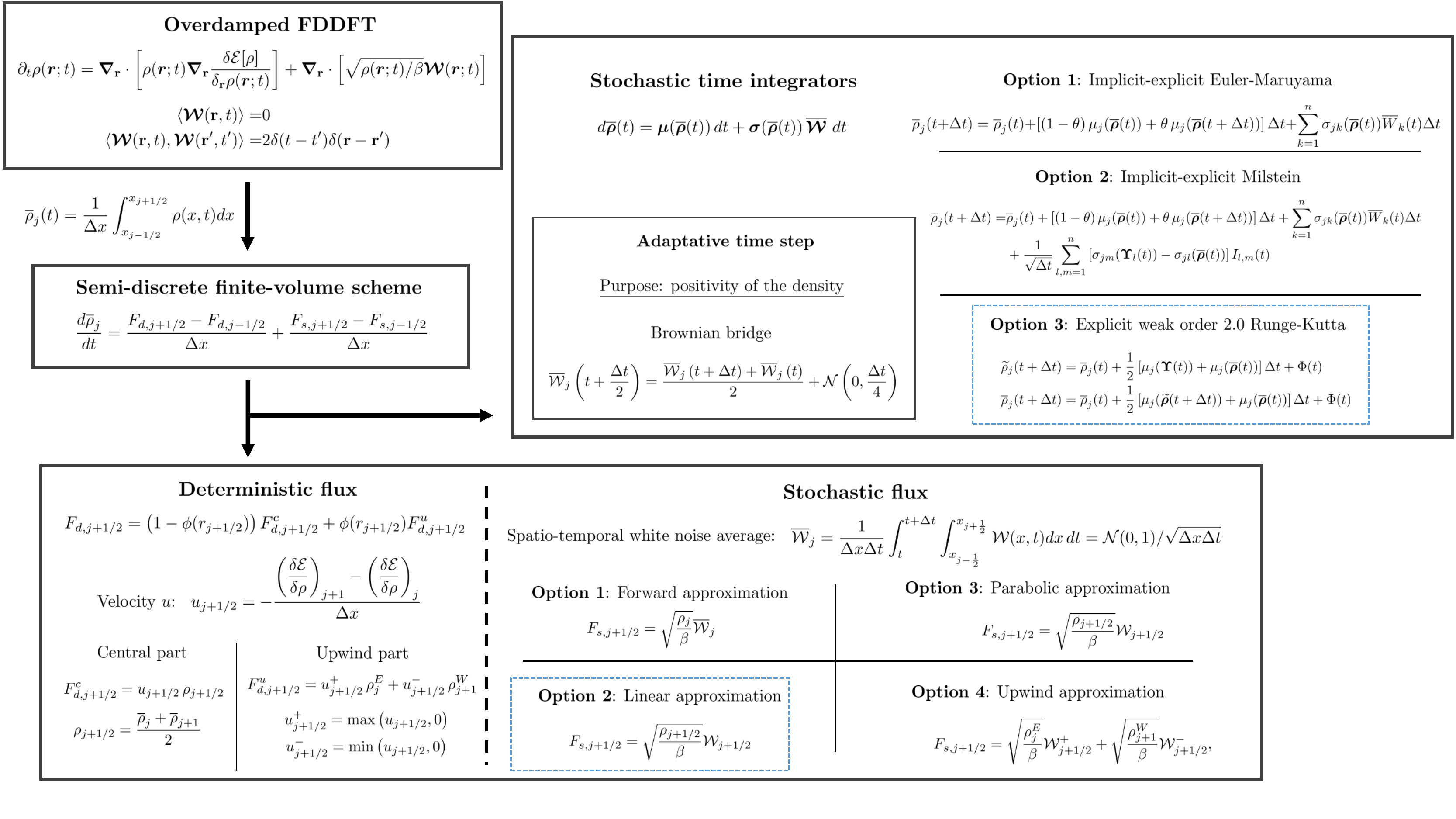}
	\end{center}
	\protect\protect\caption{\label{fig:chart}
		Flowchart reporting the main equations adopted to solve the overdamped FDDFT with a finite-volume approach.
		Arrows denote the connections among the different steps.
		Dashed boxes denote the options shown to provide the best accuracy and efficiency to simulate our SPDE \eqref{eq:FDDFTflux}, as explained in Sect.~\ref{numerical_applications_eq}.
		%
		%
	}
\end{figure}

\section{Numerical applications}
\label{numerical_applications_eq}

In this section we provide tests of the numerical schemes developed in
Sect.~\ref{sec:nummet}. Initially, in subsection \ref{subs:equil} we conduct
a simulation with a purely-diffusive ideal-gas free energy with noise and
without external fields or interparticle potentials.  There are several
theoretical results for such systems~\cite{bell2007numerical,kim2017}
allowing us to benchmark the statistical correlation and the structure factor
from our numerical schemes. Further validation of the schemes will be offered
via comparison with our own MD simulations. The results of these tests show
that the Runge-Kutta temporal integrator \eqref{eq:RK2} and the linear
approximation of the stochastic flux \eqref{eq:linear}-\eqref{eq:linear2}
accomplish the best accuracy and efficiency to simulate our SPDE
\eqref{eq:FDDFTflux}. This choice is maintained in the examples that follow.

Secondly, we provide a simulation for an ideal gas with a local confining external potential $V(x)$, in order to test the mean and variance of the density, the spatial correlation and the decay of the discrete free energy in time.

The simulations of ideal gases are also compared with results from MD
simulations using the software LAMMPS \cite{plimpton1995fast}.

\subsection{Ideal-gas system in equilibrium}\label{subs:equil}

Consider the SPDE in \eqref{eq:FDDFToverd} without any external or interaction potential ($V(x)=W(x)=0$) and applied to the classical ideal-gas free energy
\begin{equation}
\mathcal{E}[\rho]= \beta^{-1} \int \rho \left(\ln (\rho)-1\right) dx,
\end{equation}
leading to a diffusive equation with multiplicative noise of the form
\begin{equation}\label{eq:diffeq}
\partial_t \rho = \Delta \rho/\beta  + \boldsymbol{\nabla} \cdot \left[ \sqrt{ \rho/ \beta} \noi(x,\rho) \right].
\end{equation}
The initial density profile is taken as the equilibrium one, with a constant
value in all cells of $\rho_j=0.5$ and a total number of particles of
$N=1000$ for the MD simulation. The mean density profile $\overline{\rho}$
at any time, taken as the mean of the density ensemble averages at every
cell, is expected to remain as $\overline{\rho}\approx 0.5$ throughout the
simulation due to the equilibrium state. The number of cells in the domain is
$n=40$, the cell size is $\Delta x=50$ and the time step is computed as
$\Delta t=0.1 \Delta x^2$ (selected as in \cite{kim2017}), the number of
trajectories is $100$, and the number of time steps is $2000$, unless
otherwise stated. The boundary conditions are periodic
and the parameter $\beta$ is fixed at $\beta=1$.

The objective is to evaluate how the different numerical
methods perform by focusing on four different statistical properties at
equilibrium: variance, spatial correlation, time correlation and structural
factor. These tests are usually employed in the literature for the validation
of stochastic numerical schemes for
FH~\cite{bell2007numerical,kim2017,voulgarakis2009bridging}. The advantage of
testing these properties at equilibrium is that their exact theoretical
values are known and can be used for comparison purposes.
Density fluctuations of a system with fixed volume $V$ can be computed as
$\langle \delta \rho^2 \rangle = \rho^2 \langle \delta N^2 \rangle / N^2$,
where $N$ and $\langle \delta N^2 \rangle$ denote average and variance of the
number of particles in $V$, respectively.
As shown in Ref.~\cite{landau1980statistical}, the variance is given by:
\begin{align}
\langle \delta N^2 \rangle = -T \frac{\bar{N}^2}{V^2} \left(\frac{\partial V}{\partial p} \right)_T
\end{align}
where $T$ and $p$ are the temperature and pressure of the system,
respectively.
Employing the equation of state (in reduced units) for an ideal gas, $pV=NT$,
we obtain $\langle \delta N^2 \rangle = N$.
In the case of infinite systems, the fluctuations of an ideal gas are
spatially uncorrelated, namely $\langle \delta \rho_i(t) \delta \rho_j(t)
\rangle = \langle \delta \rho^2 \rangle \delta_{ij}^K$.
However, for finite systems the constraint on conservation of mass introduces correlations~\cite{bell2007numerical}.
Expressing the space correlations of density fluctuations as $\langle \delta
\rho_i(t) \delta \rho_j(t) \rangle = A \delta_{ij}^K + B$, conservation
of mass dictates $\sum_i \langle \delta \rho_i(t) \delta \rho_j(t) \rangle =
0$, which corresponds to the constraint $B= -A/n$, with $n$ being the total
number of cells.
Moreover, in the limit $n \mapsto \infty $ the fluctuations for an infinite system have to be recovered, thus $A=\langle \delta \rho^2 \rangle$.
It follows that the spacial correlation for the closed system can be expressed as:
\begin{align}
\langle \delta \rho_i(t) \delta \rho_j(t) \rangle = \langle \delta \rho^2 \rangle \left( \delta_{ij}^K - \frac{1}{n} \right) .
\label{eq:correlation}
\end{align}

The expression for the variance allows us also to obtain a quick estimation
of the minimum cell size for which, due to thermal fluctuations, negative
density values are likely to occur.
The expected value of the density fluctuations for an ideal gas can be
expressed through its standard deviation $\sqrt{ \langle \delta \rho^2 \rangle} =
\sqrt{ \rho / \Delta V}$.
Thus, with a confidence of $99.7\%$, the maximum values of the density
fluctuations will be $\sqrt{ \langle \delta \rho^2 \rangle} \mid_{\max} \sim
3 \sqrt{\frac{\rho}{ \Delta V}}$.
It follows that the noise fluctuations give negative density values with a probability higher than $0.3\%$ when the following condition
is verified:
\begin{align}
3 \sqrt{\frac{\rho}{ \Delta V}} \gtrsim \rho \quad \text{or, equivalently} \quad \Delta V \lesssim  \frac{3}{\sqrt{\rho}}
\end{align}

In subsubsection \ref{sss:choice} we offer a discussion on the accuracy and efficiency of the temporal integrator and spatial discretization, using the results from the four tests and the computational cost. The justified choices, which are the Runge-Kutta temporal integrator \eqref{eq:RK2} and linear approximation of the stochastic flux \eqref{eq:linear}-\eqref{eq:linear2}, are employed during the four tests, in the sense that the Runge-Kutta temporal integrator is employed when evaluating the different spatial discretizations, and the linear approximation of the stochastic flux is employed when evaluating the temporal integrators.

\subsubsection{Standard deviation}\label{subsub:std}

For this test we aim to evaluate how the standard deviation of the density
varies depending on the number of particles per cell $N_c$. We keep the total length
	and the total number of particles in the domain
	as constant, and  we only vary the number of particles per cell
	by enlarging or shortening the cell size $\Delta x$. Consequently this analysis helps to elucidate how changing the finite-volume lattice size affects the numerical statistical properties . The mean
density of the profile is $\overline{\rho}=N/(n\Delta x)$.

As shown above, the theoretical standard deviation of the diffusion SPDE \eqref{eq:diffeq} applied in finite systems in equilibrium satisfies
\begin{equation}\label{eq:theoretstd}
\sqrt{ \langle \delta \rho^2 \rangle}_{\text{theory}}= \frac{\langle \rho \rangle }{\sqrt{N_c}} \sqrt{ 1-\frac{1}{n} }.
\end{equation}
As a remark, in spite of the fact that $\sqrt{ \langle \delta \rho^2
	\rangle}_{\text{theory}}$ holds for all $N_c$, previous studies
\cite{donev2010,voulgarakis2009bridging} have pointed out that there should
be a minimum of 5-10 particles per cell to recover the microscopic
statistical properties by means of FH. This occurs because with such low
number of particles per cell the particle fluctuations are not accurately
modelled with the multiplicative noise in \eqref{eq:diffeq}.

The results of this study are displayed in Fig.~\ref{fig:std}, depicting a
comparison of the temporal schemes (a) and spatial discretizations (b)
against the theoretical standard deviation \eqref{eq:theoretstd} and the one
computed from MD. It is evident from both plots that all schemes
approximate correctly the standard deviation for $N_c>5$. Below this number
of particles per cell, the standard deviations deviate from the expected ones.
This result chimes in with the minimum number of 5-10 particles per cell
necessary to recover the statistical properties in FH.

There are no remarkable differences between the temporal integrators or spatial discretizations for the stochastic flux.

\begin{figure}[ht!]
\begin{center}
\subfloat[]{\protect\protect\includegraphics[scale=0.9,trim=0 15 0 0]{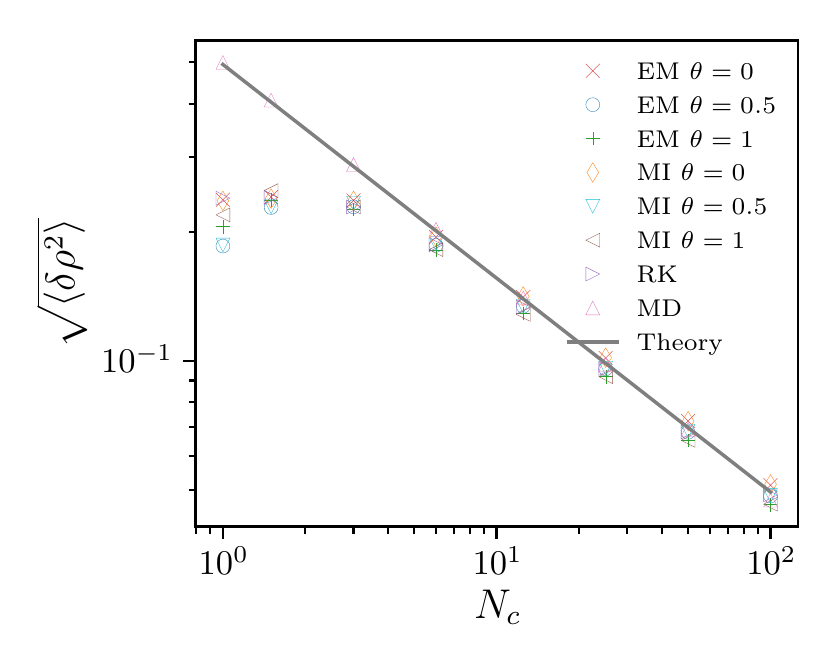}
}
\subfloat[]{\protect\protect\includegraphics[scale=0.9,trim=0 15 0 0]{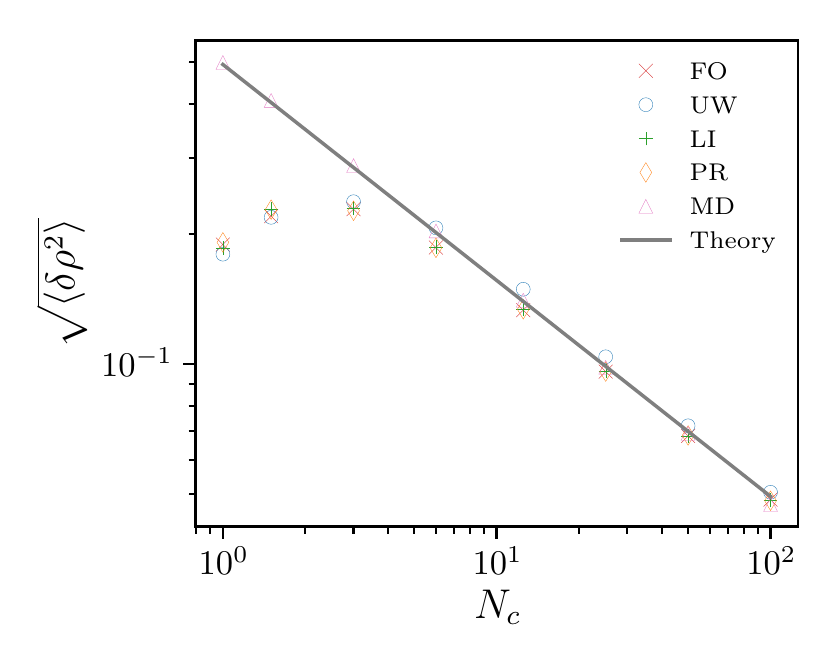}
}
\end{center}
\protect\protect\caption{\label{fig:std} Standard deviation $\sqrt{ \langle \delta \rho^2 \rangle}$ as a function of the number of particles per cell $N_c$ (or, equivalently, as a function of the cell size given that total volume and density of the system are kept constant throughout the simulations),
	for an ideal gas in equilibrium. (a) Temporal integrators. EM:
	Euler-Maruyama, MI: Milstein, RK: Runge-Kutta, MD: Molecular dynamics.
	Explicit ($\theta=0$), semi-implicit ($\theta=0.5$) and implicit
	($\theta=1$). (b) Spatial discretizations of the stochastic flux. FO: Forward
	\eqref{eq:forward}, UW: Upwind \eqref{eq:stoflux}-\eqref{eq:noipm}, LI:
	Linear \eqref{eq:linear}-\eqref{eq:linear2}, PR: Parabolic
	\eqref{eq:parabolic}-\eqref{eq:parabolic2}, Theory:
	Eq.~\ref{eq:theoretstd}.}
\end{figure}

\subsubsection{Time correlations}

The objective of this test is to measure the time correlation of the density
in one specific cell of the domain. The normalized time correlation function
is defined as
\begin{equation}\label{eq:ct}
C_T (t)=\frac{\langle \delta \rho_i(t) \delta \rho_i(0) \rangle}{\langle \delta \rho_i(0) \delta \rho_i(0) \rangle},
\end{equation}
where $\delta \overline{\rho}_i(t)=\overline{\rho}_i(t)-\overline \rho$. The
time correlation function expected to decay in time for any process in
equilibrium, including the diffusion equation \eqref{eq:diffeq}. Previous
studies \cite{bell2007numerical} have compared the numerical results with the
Fourier transform of the time correlation \eqref{eq:ct}, which is denoted as
the spectral density and for which there are explicit expressions available.
In spite of this, these exact expressions for the spectral density do not
take into account the finite-size effects in the numerical simulations,
leading to a lack of agreement in the results \cite{bell2007numerical}.

We have then decided to compare the results obtained from the numerical
schemes in Sect.~\ref{sec:nummet} with MD simulations only, which indeed take
into account the finite-size effects. The results
are displayed in Fig.~\ref{fig:timecor}. For all schemes we evidence a clear
decay in time of the time correlation. Concerning the temporal integrators,
the explicit ones ($\theta=0$) tend to be closer to the MD simulations initially, while the implicit ones ($\theta=1$) provide a
better approximation in the long-time regimes. With respect to the spatial
discretizations for the stochastic flux, the upwind one deviates the most
from MD, while the rest of them behave similarly.

\begin{figure}[ht!]
\begin{center}
\subfloat[]{\protect\protect\includegraphics[scale=0.9,trim=0 15 0 0]{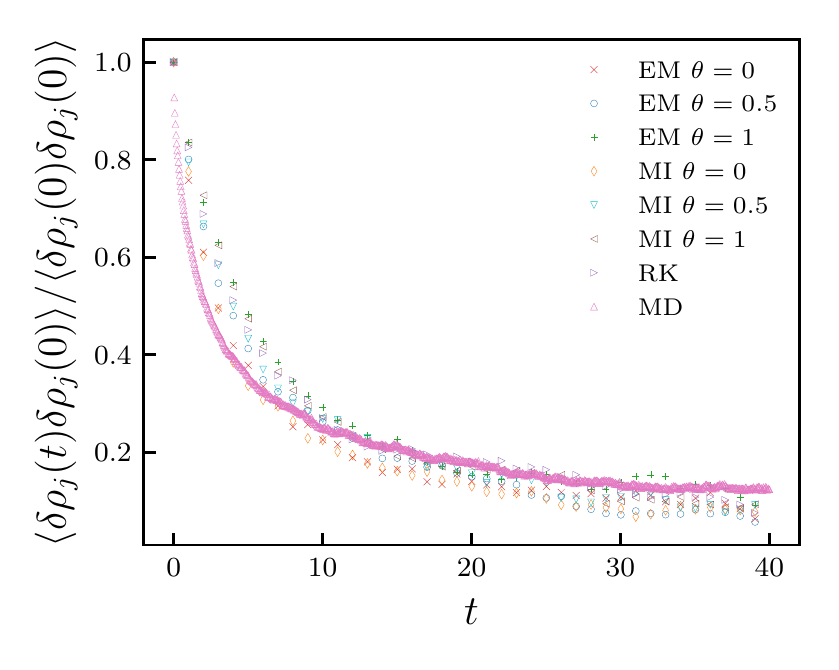}
}
\subfloat[]{\protect\protect\includegraphics[scale=0.9,trim=0 15 0 0]{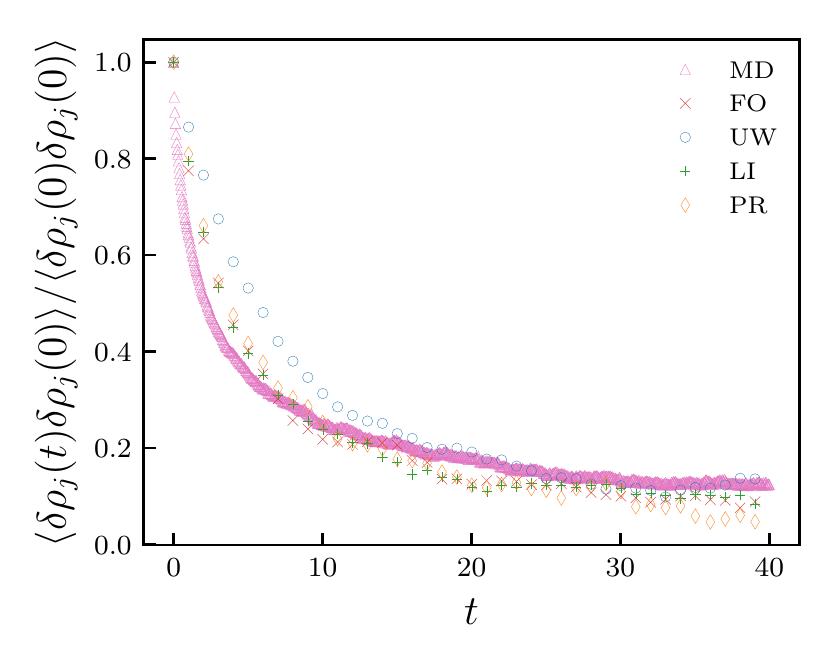}
}
\end{center}
\protect\protect\caption{\label{fig:timecor} Temporal decay of the normalized time correlation $C_T$,
	defined as in \eqref{eq:ct}, for an ideal-gas system in equilibrium (a)
	Temporal integrators. EM: Euler-Maruyama, MI: Milstein, RK: Runge-Kutta, MD:
	Molecular dynamics. Explicit ($\theta=0$), semi-implicit ($\theta=0.5$) and
	implicit ($\theta=1$). (b) Spatial discretizations of the stochastic flux.
	FO: Forward \eqref{eq:forward}, UW: Upwind
	\eqref{eq:stoflux}-\eqref{eq:noipm}, LI: Linear
	\eqref{eq:linear}-\eqref{eq:linear2}, PR: Parabolic
	\eqref{eq:parabolic}-\eqref{eq:parabolic2}.}
\end{figure}

\subsubsection{Spatial correlations}

This test seeks to evaluate whether the proposed numerical schemes in Sect.~\ref{sec:nummet} satisfy the exact spatial correlation for finite-size systems derived above:
\begin{align}\label{eq:csexact}
\langle \delta \rho_i(t) \delta \rho_j(t) \rangle = \frac{ \langle \rho \rangle }{\Delta x} \left( \delta_{ij} - \frac{1}{n} \right).
\end{align}
Contrary to the infinite-domain case where there are no spatial correlations
between adjacent cells, for the finite-size case there is an extra term $1/n$
which decreases as the number of cells $n$ increases.

The results of this test are depicted in Fig.~\ref{fig:spacecor}, with the normalized spatial correlation
\begin{equation}\label{eq:cs}
C_S (t)=\frac{\langle \delta \rho_i(t) \delta \rho_j(t) \rangle}{\langle \delta \rho_i(0) \delta \rho_i(0) \rangle}
\end{equation}
with $\delta \overline{\rho}_i(t)=\overline{\rho}_i(t)-\overline \rho$,
plotted for each of the numerical schemes, the MD simulations and the exact
expression \eqref{eq:csexact}. The main conclusion is that most of the
temporal integrators and spatial discretizations approximate adequately the
theoretical expression \eqref{eq:csexact}, as it is depicted in
Fig.~\ref{fig:spacecor}. The fully explicit and implicit Euler-Maruyama and
Milstein slightly deviate with respect to the theoretical spatial correlation
in the cells adjacent to the central cell, while the semi-implicit schemes
perform correctly.

\begin{figure}[ht!]
\begin{center}
\subfloat[]{\protect\protect\includegraphics[scale=0.9,trim=0 15 0 0]{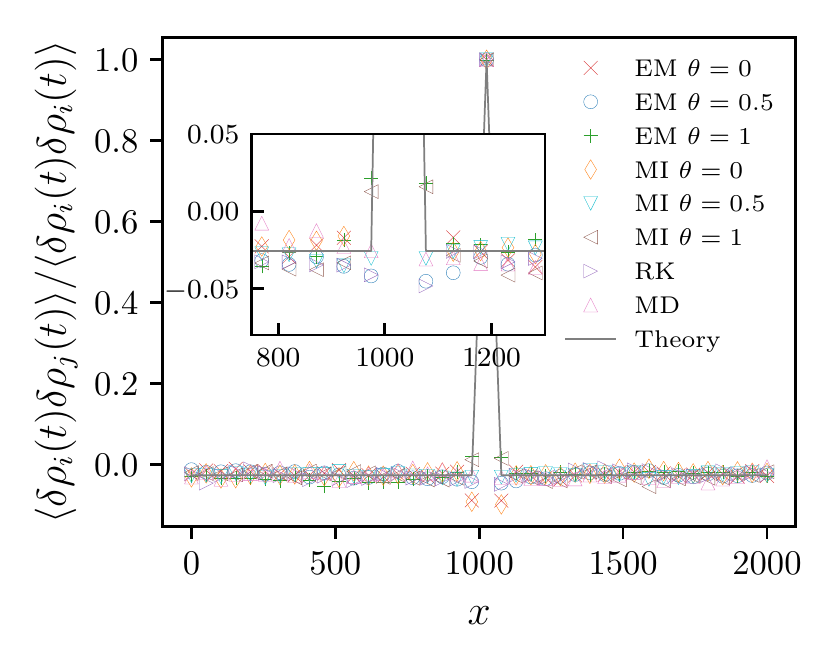}
}
\subfloat[]{\protect\protect\includegraphics[scale=0.9,trim=0 15 0 0]{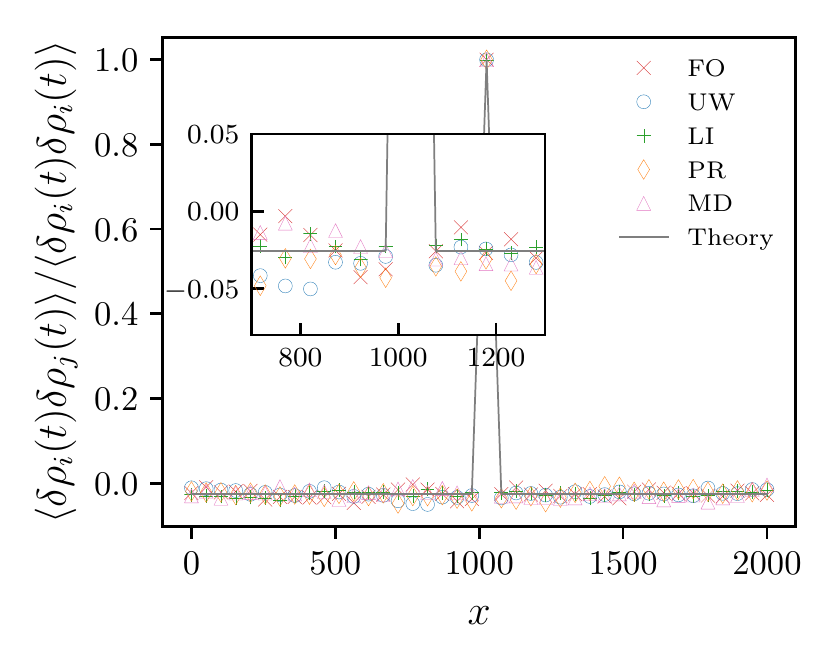}
}
\end{center}
\protect\protect\caption{\label{fig:spacecor} Normalized spatial correlation \eqref{eq:cs} for an ideal-gas system in equilibrium.
(a) Temporal integrators. EM: Euler-Maruyama, MI: Milstein, RK: Runge-Kutta,
MD: Molecular dynamics. Explicit ($\theta=0$), semi-implicit ($\theta=0.5$)
and implicit ($\theta=1$). (b) Spatial discretizations of the stochastic
flux. FO: Forward \eqref{eq:forward}, UW: Upwind
\eqref{eq:stoflux}-\eqref{eq:noipm}, LI: Linear
\eqref{eq:linear}-\eqref{eq:linear2}, PR: Parabolic
\eqref{eq:parabolic}-\eqref{eq:parabolic2}.}
\end{figure}

\subsubsection{Structure factor}

This test evaluates how the structure factor $S$ at equilibrium is
approximated by the temporal and spatial discretizations. Even though its
general form satisfies \eqref{eq:structurefgeneral}, its theoretical
expression for an ideal gas without external potential is given by
\eqref{eq:structuref}, so that for the current numerical simulation with
$\beta=1$ it follows that $S/\left<\rho\right>=1$.

The discrete structure factor is computed from Eqs~\eqref{eq:st1}-\eqref{eq:st2}. First the discrete spatial Fourier transform of the density satisfies:
\begin{equation}
\hat{\rho}(\lambda)=\frac{1}{n}\sum_j \overline{\rho}_j e^{-i \lambda x_j}.
\end{equation}
Subsequently, the structure factor follows from
\begin{equation}\label{eq:structurefdiscrete}
S(\lambda)=\frac{\left<\delta \hat{\rho}(\lambda)\, \delta \hat{\rho}^*(\lambda)\right>}{n \Delta x},
\end{equation}
where $\delta \hat{\rho}(\lambda)= \hat{\rho}(\lambda)-\left< \hat{\rho}(\lambda)\right>$ and $\hat{\rho}^*$ denotes the complex conjugate of $\hat{\rho}$.

The results of this test for the structure factor at equilibrium are depicted
in Fig.~\ref{fig:structure_factor}. The theoretical value of the structure
factor, along with the performed MD simulations, allows us to judge whether
the temporal integrators and spatial discretizations perform accurately. On
the one hand, from Fig.~\ref{fig:structure_factor}(a) it is evident how
the explicit Euler-Maruyama and Milstein temporal integrators overestimate
the structure factor for large $\lambda$, while their implicit versions
underestimate it for large $\lambda$ too. The semi-implicit schemes and the
Runge-Kutta behave correctly, and the small damping in the numerical
structure factor for all $\lambda$ is due to the choice of the hybrid
deterministic flux, as it was explained from
Fig.~\ref{fig:deterministic_flux}. On the other hand, from the spatial
discretizations of the stochastic flux there is a clear deviation when
applying the upwind form. In addition, the forward discretization seems to
slightly oscillate for lower $\lambda$. The rest of discretizations
approximate the theoretical value correctly, with the small damping already
mentioned.

\begin{figure}[ht!]
	\begin{center}
		\subfloat[]{\protect\protect\includegraphics[scale=0.9,trim=0 15 0 0]{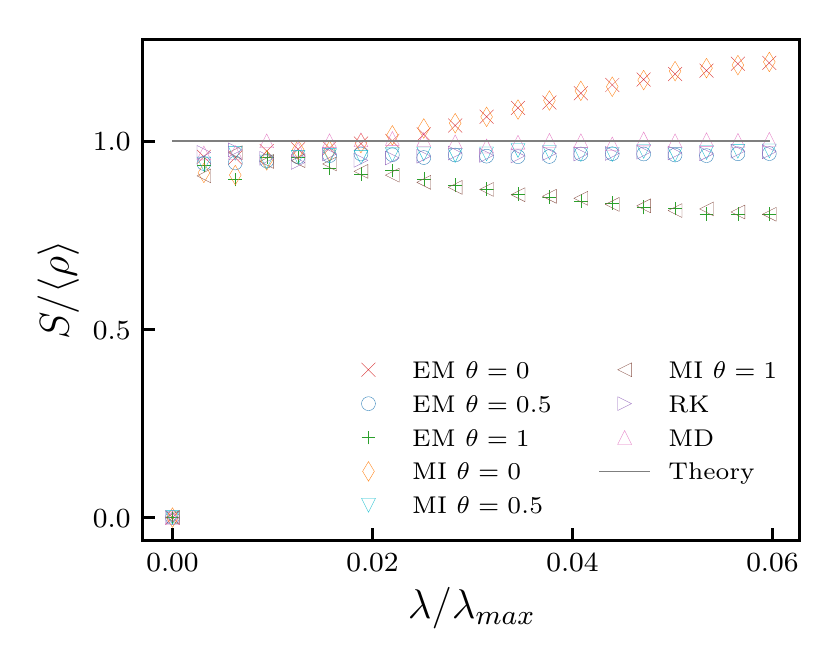}
		}
		\subfloat[]{\protect\protect\includegraphics[scale=0.9,trim=0 15 0 0]{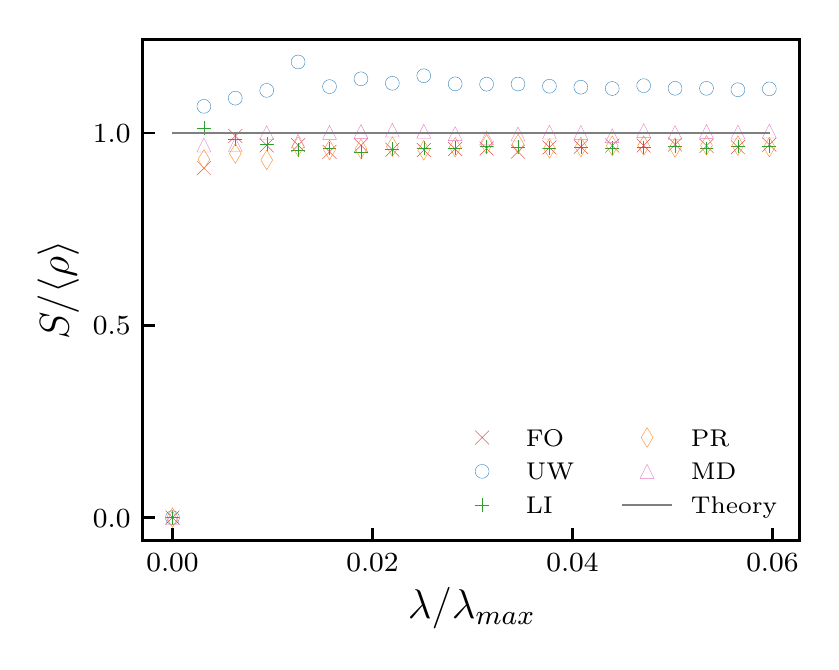}
		}
	\end{center}
		\protect\protect\caption{\label{fig:structure_factor} Structure factor \eqref{eq:structurefdiscrete} for an
		ideal-gas system in equilibrium. (a) Temporal integrators. EM:
		Euler-Maruyama, MI: Milstein, RK: Runge-Kutta, MD: Molecular dynamics.
		Explicit ($\theta=0$), semi-implicit ($\theta=0.5$) and implicit
		($\theta=1$). (b) Spatial discretizations of the stochastic flux. FO: Forward
		\eqref{eq:forward}, UW: Upwind \eqref{eq:stoflux}-\eqref{eq:noipm}, LI:
		Linear \eqref{eq:linear}-\eqref{eq:linear2}, PR: Parabolic
		\eqref{eq:parabolic}-\eqref{eq:parabolic2}.}
\end{figure}

\subsubsection{Temporal integrators and spatial discretization of the stochastic flux}\label{sss:choice}

With respect to the temporal integrators, both the fully explicit and
implicit Euler-Maruyama and Milstein present certain deviations in the time
correlation (Fig.~\ref{fig:timecor}), spatial correlation
(Fig.~\ref{fig:spacecor}) and structure factor
(Fig.~\ref{fig:structure_factor}). Their semi-implicit versions and the
Runge-Kutta behave similarly in all tests, and approximate adequately the
theoretical and MD results. Their relative costs are compared by means of
Fig.~\ref{fig:cpu_time}. While the cost of the Runge-Kutta escalates with order
$\mathcal{O}(n^{3})$, the cost of the semi-implicit Euler-Maruyama and
Milstein has an order of $\mathcal{O}(n^{2})$. However, due to the different
constant coefficient in the cost, the plot clearly shows that for $n<100$ the
Runge-Kutta cost is lower than that of  semi-implicit schemes, while for $n>100$
it is higher.

The Milstein scheme, which guarantees a higher strong order
		convergence, was tested because in previous works it performed well in
		conjunction with adaptive time-step algorithms based on Brownian trees
		\cite{lutsko2015two}.
		However, from the simulation results, we observed that the higher
		computational cost of this numerical method did not lead to a increased
		accuracy compared to the implicit Euler-Maruyama and to the weak Runge-Kutta
		schemes. 
Because of these reasons and together with the fact that in the
simulations of this work $n<100$, we select the Runge-Kutta temporal
integrator.

Concerning the spatial discretization of the stochastic flux, the upwind
choice does not approximate well the time correlation and structure factor,
while the forward approximation presents some deviation in the structure
factor for small $\lambda$. Hence the best choices are the linear and
parabolic approximations, which behave similarly in all test cases. We select
the linear approximation due to its lower cost since it only depends on the
density and white noise cell averages of two cells and not four.

\subsection{Ideal-gas system out of equilibrium}
For this example we consider a free energy which includes the effects of a double-well external potential, so that
\begin{align}\label{eq:freeenergyoutequil}
\mathcal{E}[\rho]  = \int \rho/\beta \left(\ln (\rho)-1\right) dx +\int V( x ) \rho\, d x,
\end{align}
and the shape of the external potential satisfies
\begin{equation}\label{eq:doublewell}
V(x)=5\left[\left(\frac{x}{n \Delta x /2}\right)^4-\left(\frac{x}{n \Delta x /2}\right)^2\right].
\end{equation}
Numerical simulations for deterministic gradient-flow equations with the free energies of the form \eqref{eq:freeenergyoutequil}-\eqref{eq:doublewell} have already been provided in \cite{carrillo2015}.
Here the objective is to evaluate how the numerical scheme in
Sect.~\ref{sec:nummet} for the FDDFT \eqref{eq:FDDFTflux} with the free
energy \eqref{eq:freeenergyoutequil}-\eqref{eq:doublewell} performs by
comparing with MD simulations. We also include a comparison with the
corresponding deterministic DDFT, which is obtained in the mean-field limit
(the most-likely path of FDDFT as noted in the Introduction).

The simulation is performed in a mesh where the number of cells is $n=40$,
each of them with width $\Delta x=5$. The time steps are $\Delta t=1$ and the
number of time steps is $n_t=2000$. The ensemble averages are computed from a
number of trajectories of $n_{traj}=1000$. We select $\beta=1$.
The MD simulation is performed by simulating $N=200$ particles, while the
deterministic DDFT simulation applies the numerical scheme in
\cite{carrillo2015} for gradient-flow equations.

The results are depicted in Fig.~\ref{fig:noneq_spacecor}.
Fig.~\ref{fig:noneq_spacecor} (a) displays the ensemble average of the
density profile at different times. The three simulations provide similar
results and we can conclude that the three approaches are comparable when
evaluating the ensemble average profile. Concerning the standard deviation
results in Fig.~\ref{fig:noneq_spacecor} (b), we find that FDDFT matches with
MD and the theoretical results in \eqref{eq:theoretstd}, while DDFT, being
deterministic, presents zero standard deviation. As already mentioned, the
FDDFT values of the standard deviation are slightly lower than the MD and
theoretical ones due to the choice of the deterministic flux in a similar
fashion to Figs.\ref{fig:deterministic_flux}-\ref{fig:std}.
Fig.~\ref{fig:noneq_spacecor} (c) shows the spatial correlation computed as
in Eqs~\eqref{eq:csexact}-\eqref{eq:cs}, with the MD and FDDFT results
approximating correctly the finite-size theoretical expression in
\eqref{eq:csexact}. DDTF does not have any spatial correlation due to the
lack of fluctuations. Finally, in \ref{fig:noneq_spacecor} (d) the temporal
evolution of the free-energy functional depending on the ensemble average
density is plotted. For the DDFT case one can appreciate that there is decay
at all times, while for MD and FDDFT there are short increases of the free
energy triggered by the fluctuations, in spite of the fact that during the
evolution there is a general decay in the free energy.

\begin{figure}[ht!]
	\begin{center}
		\subfloat[]{\protect\protect\includegraphics[scale=0.9,trim=0 15 0 0]{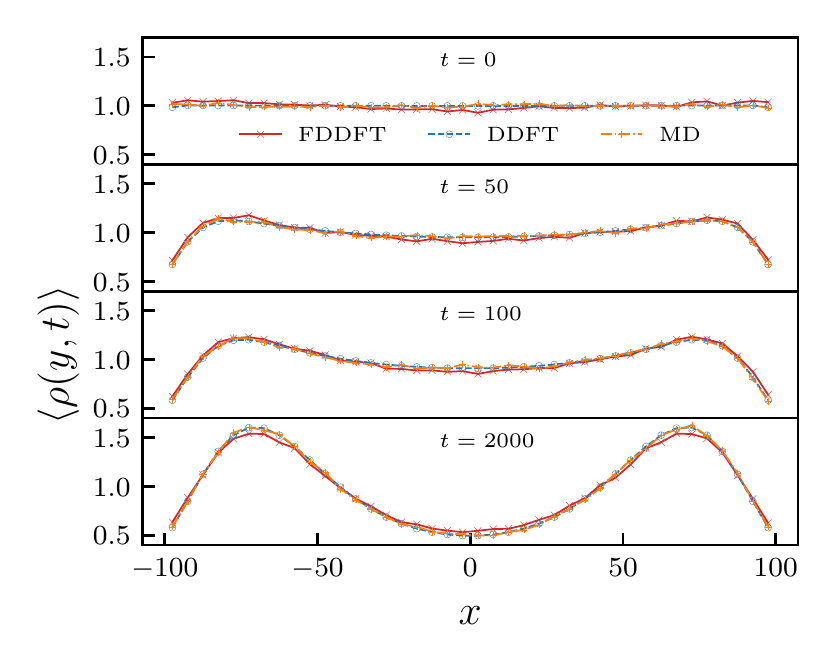}
		}
		\subfloat[]{\protect\protect\includegraphics[scale=0.9,trim=0 15 0 0]{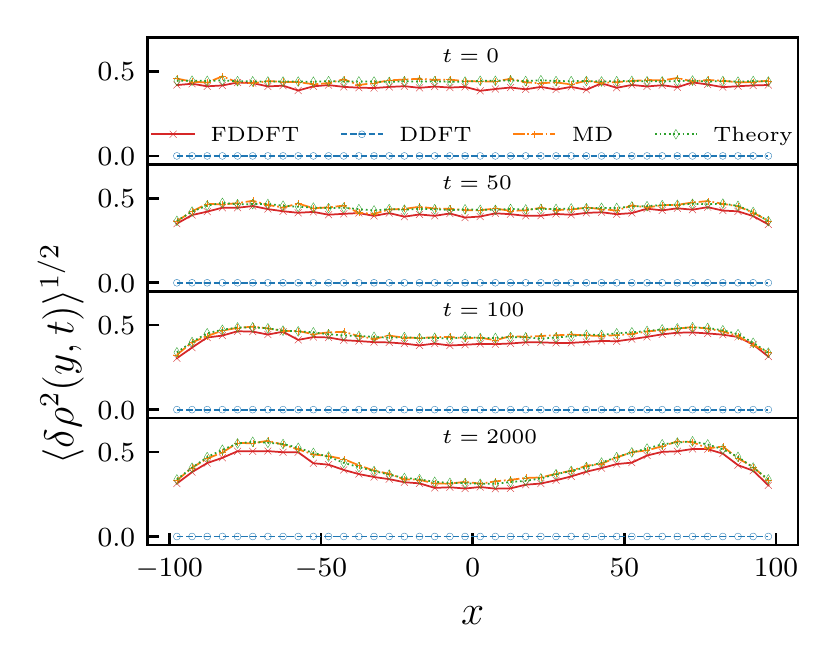}
		}
	
		\subfloat[]{\protect\protect\includegraphics[scale=0.9,trim=0 15 0 0]{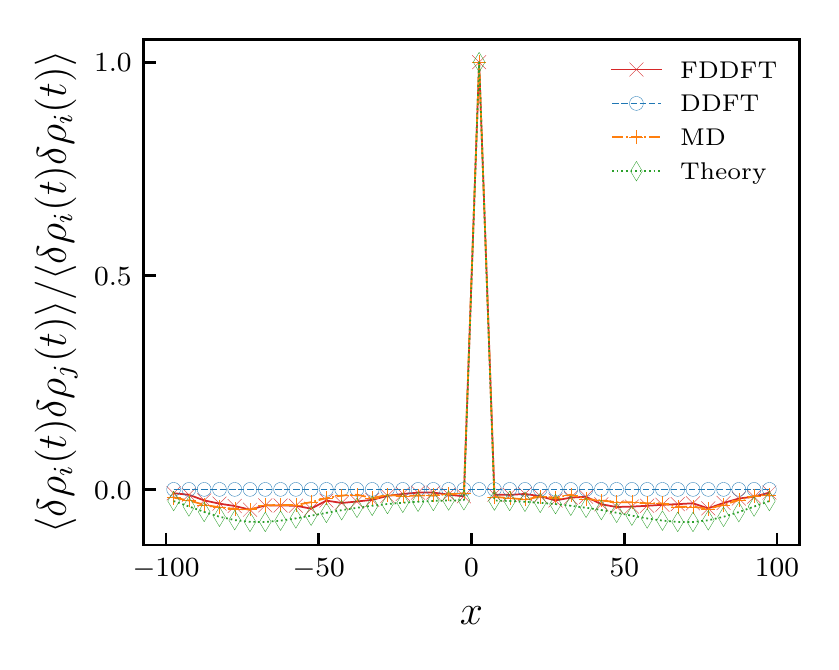}
	}
		\subfloat[]{\protect\protect\includegraphics[scale=0.9,trim=0 15 0 0]{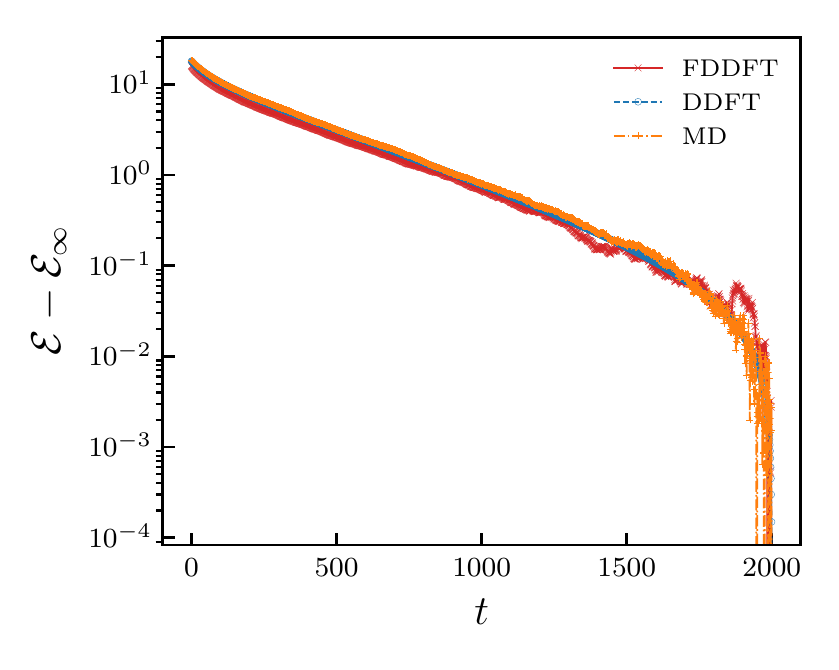}
	}
	\end{center}
	\protect\protect\caption{\label{fig:noneq_spacecor} Time evolution of mean density (a) and density standard deviation (b) fields computed with FDDFT, DDFT and MD simulations.
	A comparison in terms of steady state spatial correlations is reported in (c).
	In (d), we report the decrease in time of the energy functional of the mean density.
}
\end{figure}

\subsection{Homogeneous nucleation in Lennard-Jones systems}

\begin{figure}[t!]
	\begin{center}
		\subfloat[]{ \protect\protect\includegraphics[scale=0.67]{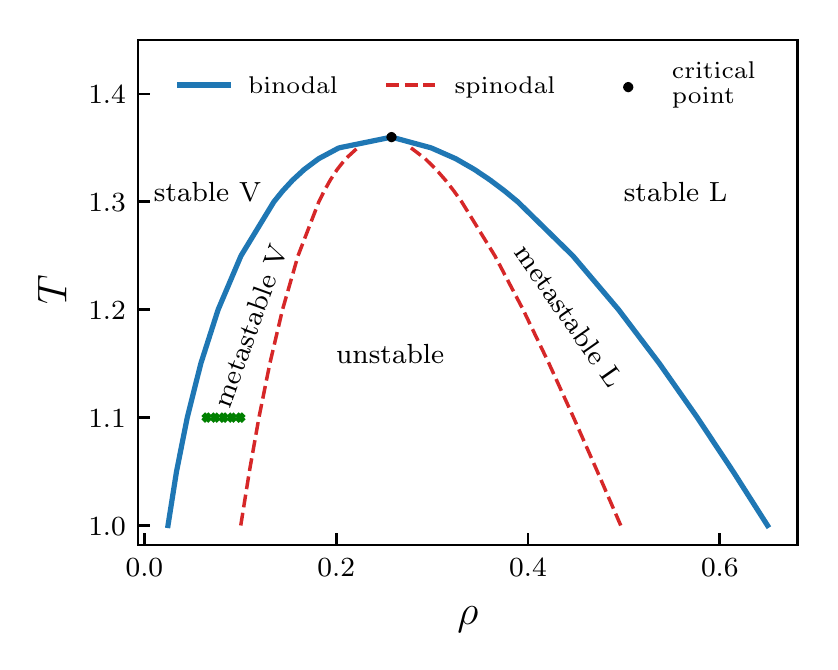}
		}
		\subfloat[]{ \protect\protect
			\includegraphics[scale=0.67]{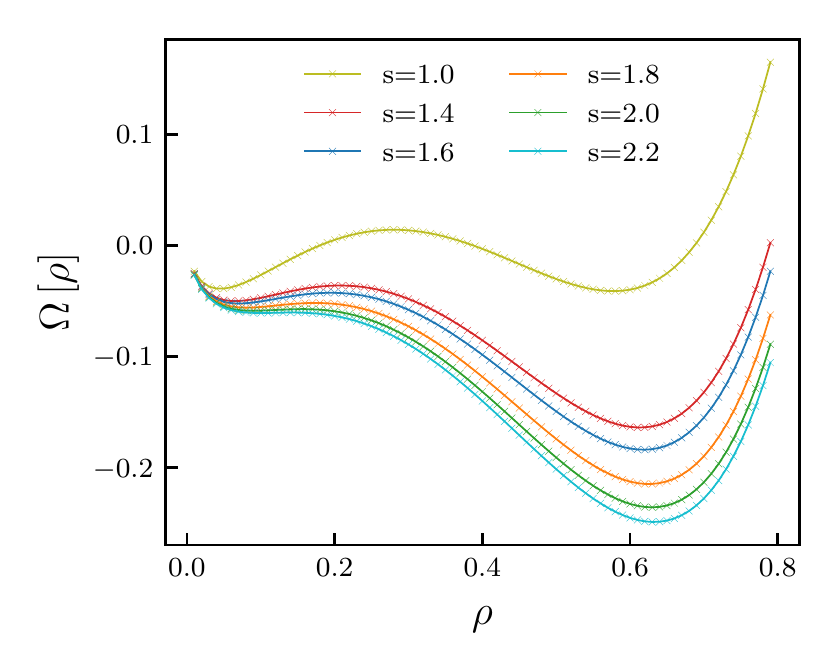}
		}
		\subfloat[]{ \protect\protect
			\includegraphics[scale=0.67]{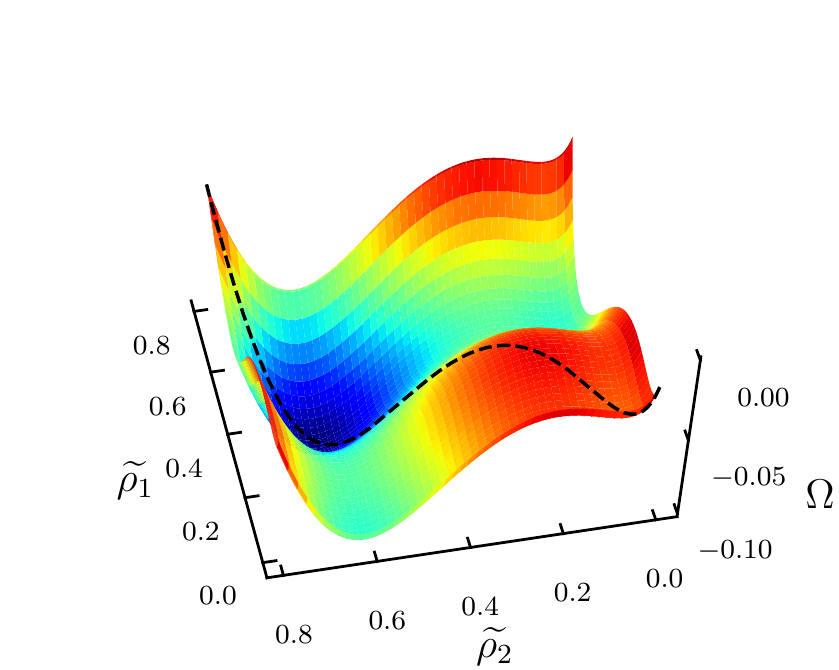}
		}
	\end{center}
	\protect\protect\caption{\label{fig:nucleation1}
		(a) we report the bulk phase diagram for the discretized LJ system.
		(b) shows the grand free-energy landscape as function of the system density for some supersaturation ratios adopted in this study.
		(c) we provide an example of free-energy landscape for systems with a non-uniform density field, with two varying densities $\widetilde{\rho}_1$ and $\widetilde{\rho}_2$.
		The dotted black line denotes the curve corresponding to bulk uniform systems.
	}
\end{figure}

\begin{figure}[t!]
	\begin{center}
		$s=1.4$\\
		\includegraphics[scale=0.73]{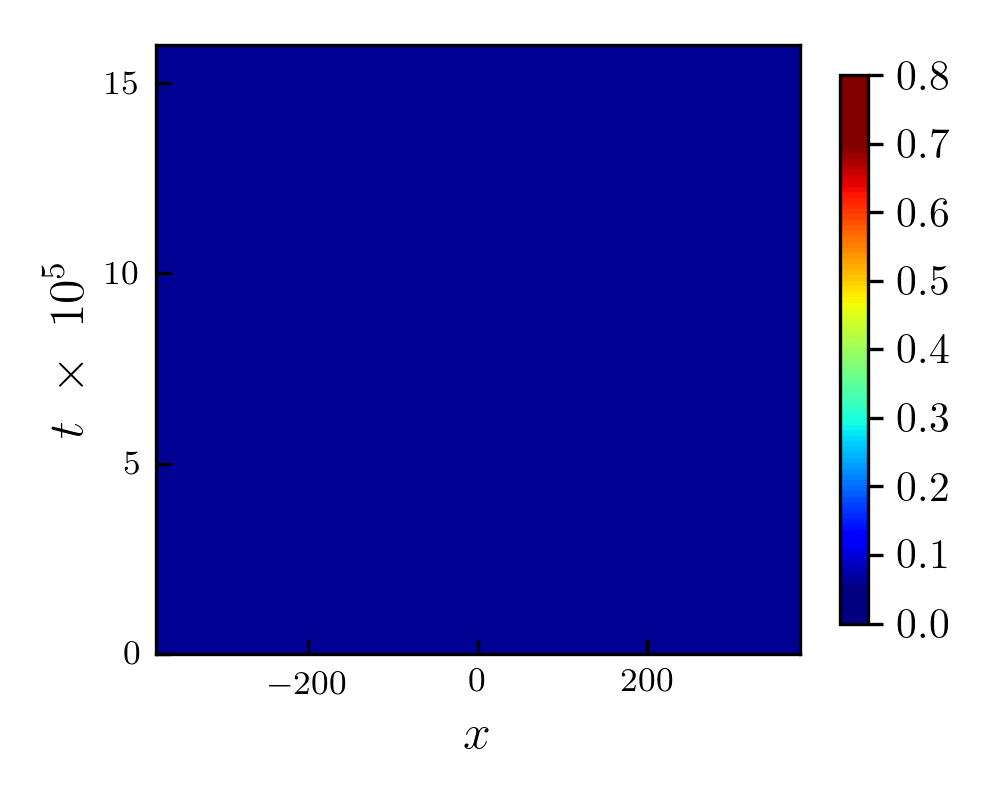}
		\includegraphics[scale=0.73]{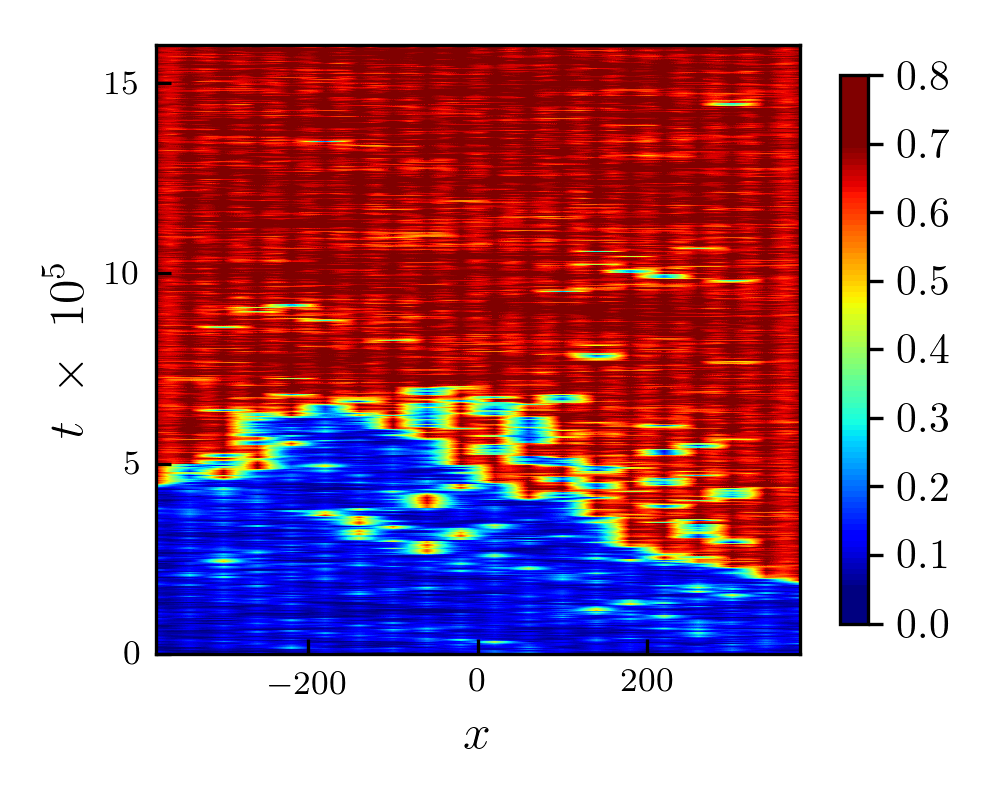}
		
		$s=1.8$\\
		\includegraphics[scale=0.73]{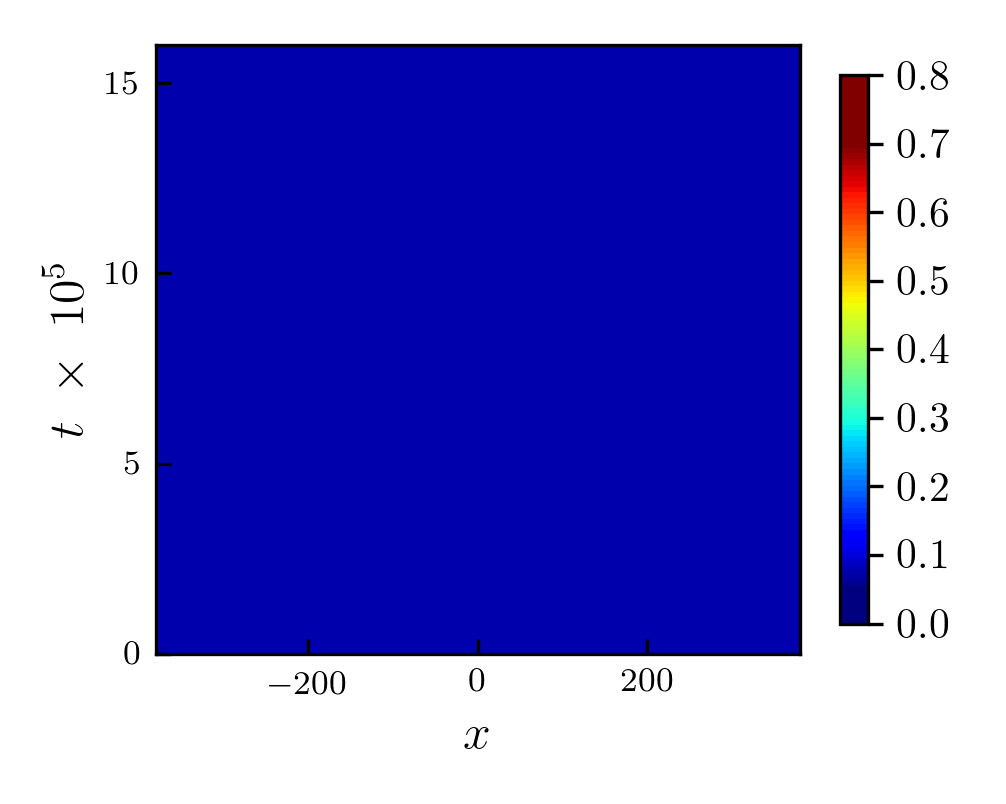}
		\includegraphics[scale=0.73]{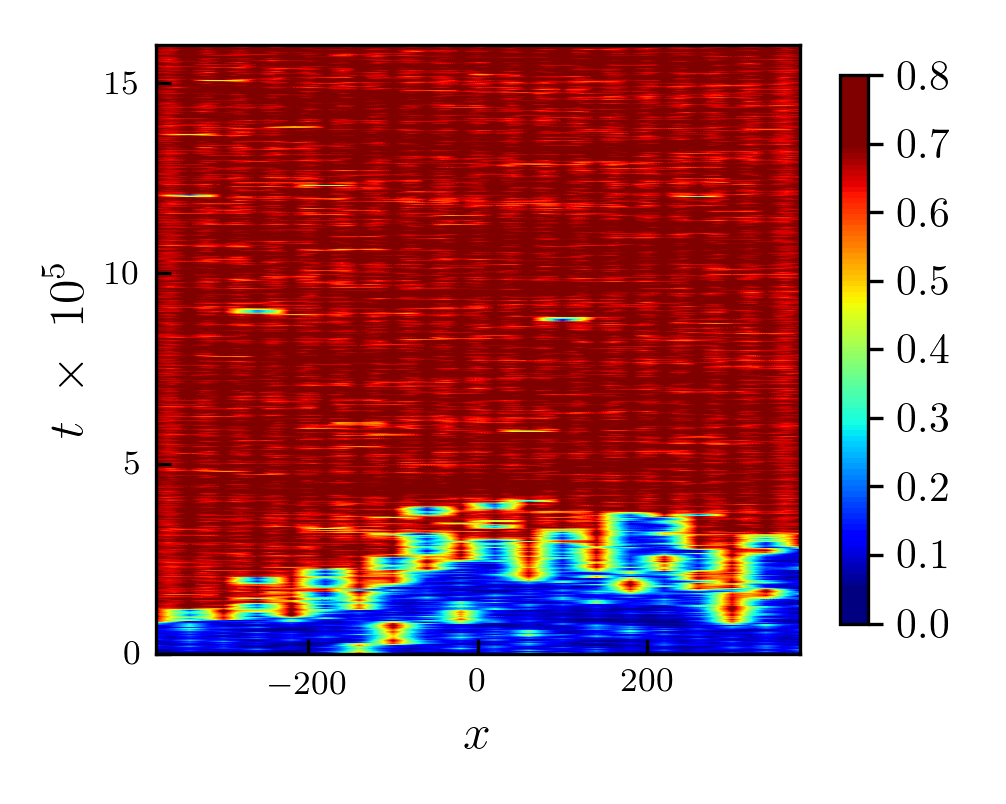}
		
		$s=2.2$\\
		\includegraphics[scale=0.73]{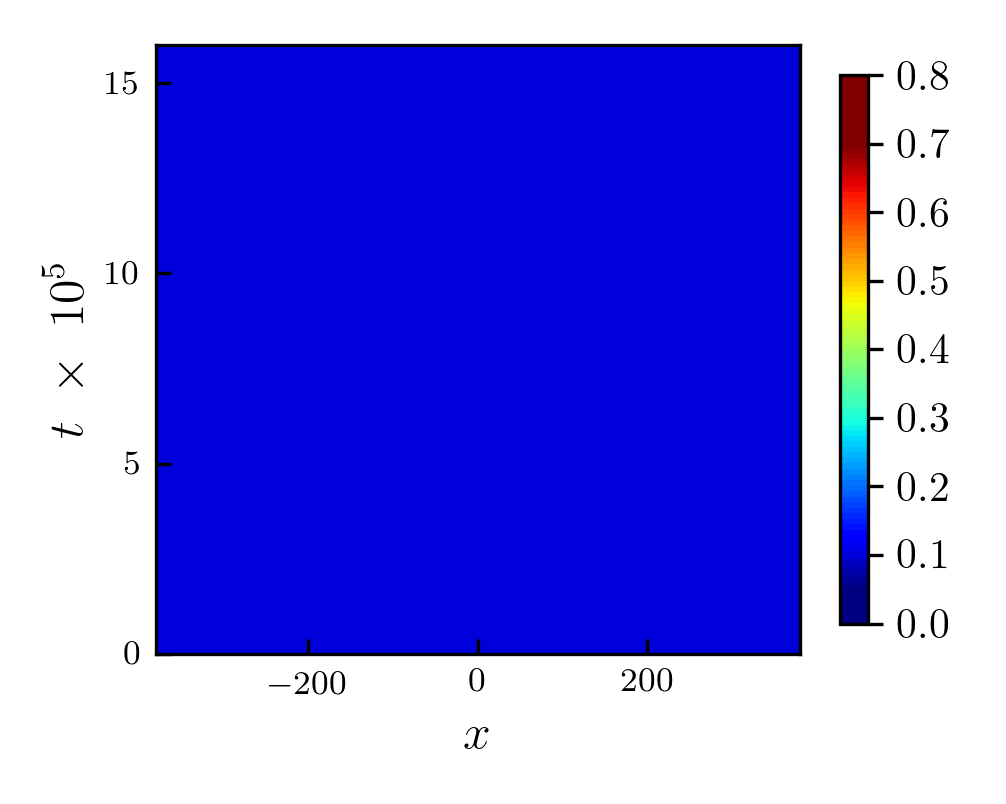}
		\includegraphics[scale=0.73]{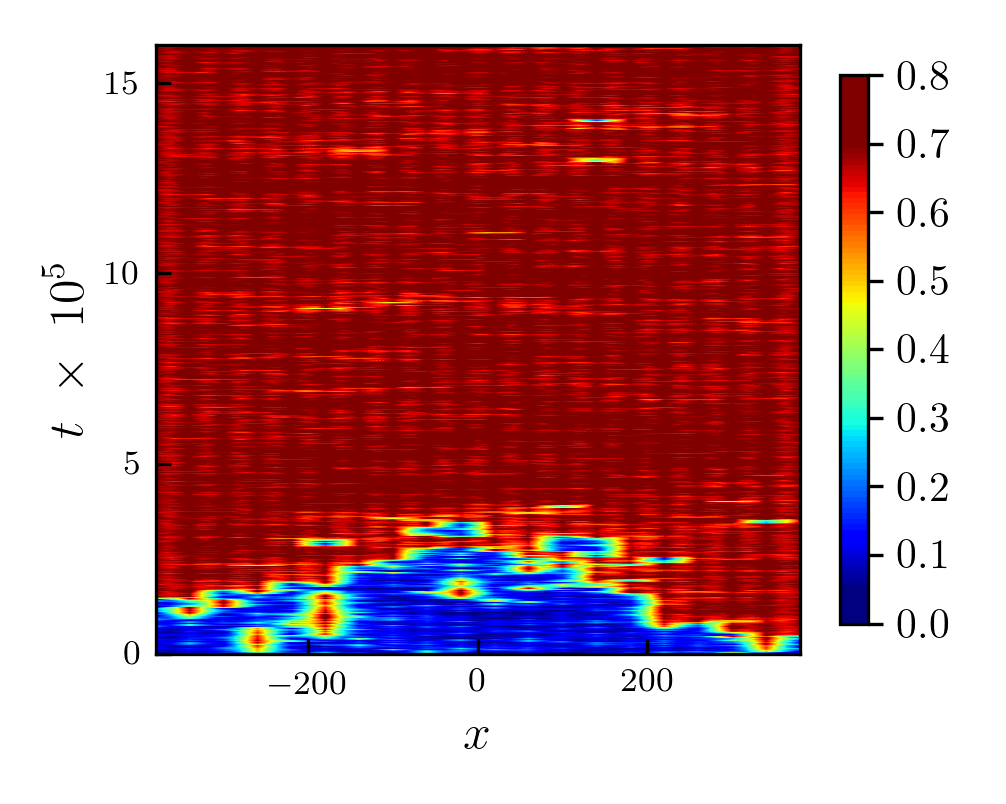}
		
	\end{center}
		\protect\protect\caption{\label{fig:nucleation2} Homogeneous nucleation of a vapour LJ system in metastable conditions with supersaturation ratio $s$.
		Left column: the mean-field evolution. Right column: a single realization of the stochastic dynamics.
	}
\end{figure}

\begin{figure}[t!]
	\begin{center}
		\subfloat[]{ \includegraphics[scale=0.65]{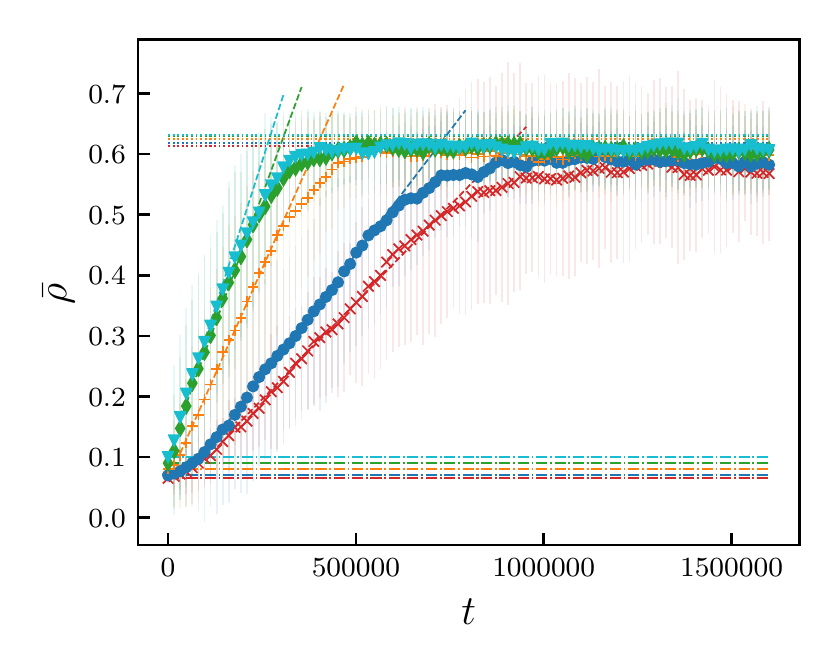}}
		\subfloat[]{ \includegraphics[scale=0.65]{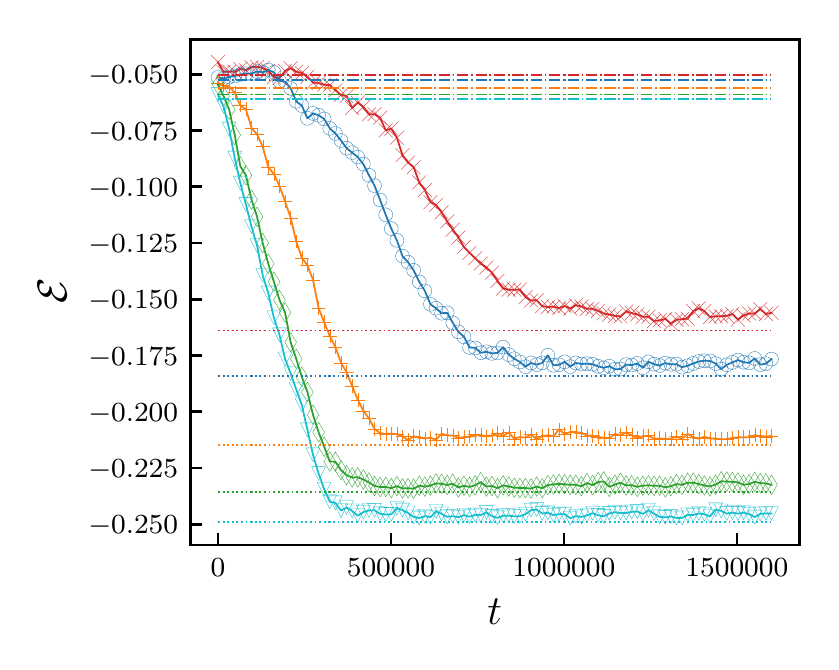} }
		\subfloat[]{ \includegraphics[scale=0.65]{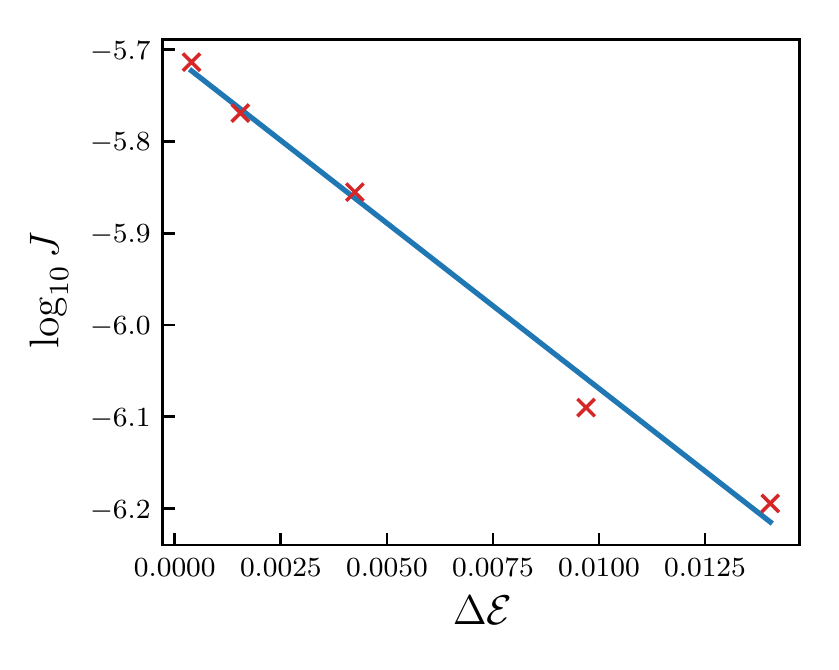} }
		
	\end{center}
\protect\protect\caption{(a) reports the evolution in time of the average system density for the supersaturation ratios adopted in this study.
(b) shows the time-evolution of the system free energy for the supersaturation ratios adopted in this study.
(c) the nucleation growth rate is plotted against the free-energy barrier.
\label{fig:nucleation3}
}
\end{figure}

The importance of fluctuations during phase transitions is crucial when considering
the homogeneous vapour-liquid transition of a Lennard-Jones (LJ) fluid.
Within the framework of DFT, the fluid density profiles of a 1D
open system exchanging particles with a reservoir at constant
temperature and chemical potential $\mu$, can be obtained from an
unconstrained numerical minimization of the grand free-energy functional
\begin{align}
\Omega [\rho (x)]  = \mathcal{F}[\rho (x) ] + \int \left( V(x) - \mu \right) \rho(x) \ dx.
\end{align}
In general, $\mathcal{F}[\rho (x) ]$ is not analytically obtainable from
first principles, except in few cases, i.e. ideal gases and hard-sphere
fluids.
In the remaining cases, $\mathcal{F}[\rho (x) ]$ is either numerically
obtained from atomistic simulations or is approximated by means of
perturbation expansions around a known free energy~\cite{Marconi1999}.
Similarly to previous works on DFT~\cite{Yatsyshin2015,Yatsyshin2015_1}, we
approximate $\mathcal{F}[\rho (x) ]$ of an LJ fluid according to the
first-order Barker-Henderson perturbation theory expansion around the
hard-sphere fluid free energy~\cite{Barker1967}, namely as
\begin{align}
\mathcal{F}[\rho (x)]  = \int \left \{ f_{\text{ID}}[\rho(x)]  +  \rho(x) f_{\text{HS}}(\rho(x)) \right \}  dx + \frac{1}{2} \int\int \rho(x) \rho(x') W(x,x') \ dx \ dx',
\end{align}
where $f_{\text{ID}}$, $f_{\text{HS}}$ and $W(x,x')$ denote ideal-gas,
hard-sphere repulsive interactions and LJ attractive contributions,
respectively.
The free energy of an ideal gas is given by
\begin{align}
f_{\text{ID}}[\rho(x)]=   k_{\text{B}} T \rho \left(  \ln( \lambda^3 \rho)- 1 \right),
\end{align}
where $\lambda$ is the thermal de Broglie wavelength.
The hard-sphere free-energy density $f_{\text{HS}}$ is obtained from the Carnahan-Starling equation of state for the hard sphere fluid, which reads~\cite{Carnahan1969}
\begin{align}
f_{HS}(\rho(x))=   k_{\text{B}} T \left(  \dfrac{  4 \eta - 3 \eta^2 }{\left(1-\eta\right)^2} \right), \quad \text{with} \quad \eta=\frac{\pi}{6}\rho \sigma^3
\end{align}
where $\sigma$ is the hard-sphere diameter set to unity in this work.
Finally, the LJ (attractive) contributions are taken into account by the following expression:
\begin{align}
W(x,x')=\begin{cases}
-1.2  \ \pi \ \epsilon \quad &\text{if} \quad \mid x-x' \mid \leq 1, \\
\pi \ \epsilon \ \left( 0.8 \ \mid x-x' \mid^{-10} -2 \ \mid x-x' \mid^{-4} \right) \quad &\text{otherwise},
\end{cases}
\end{align}
derived by integrating along $y$ and $z$ the 12-6 LJ potential~\cite{Yatsyshin2015_1}.

In order to analyse the vapour-to-liquid (first-order) phase transitions, we
first compute the coexisting density profiles.
The coexisting values of vapour and liquid density (binodal line) are denoted as $\rho_{v}$ and $\rho_{l}$ respectively, and are obtained by solving the following system of equations:
\begin{align}
\begin{cases}
& \left. \frac{ \partial \Omega} { \partial \rho } \right|_{ \rho_{v} }   = \left. \frac{ \partial \Omega} { \partial \rho } \right|_{  \rho_{l} } =0 , \\
& \Omega \left[ \rho_v \right] - \Omega \left[ \rho_l \right]  =0.
\end{cases}
\end{align}
The metastable regions are delimited by the binodal and spinodal lines.
The spinodal lines correspond to the inflection points of the grand free energy, hence are evaluated by solving:
\begin{align}
& \left. \frac{ \partial^2 \Omega} { \partial \rho^2 } \right|_{  \rho_{v} }   = \left. \frac{ \partial^2 \Omega} { \partial \rho^2 } \right|_{  \rho_{l} } =0.
\end{align}
Finally, the bulk critical point is given by the intersection between binodal and spinodal lines, and it is thus computed as
\begin{align}
& \left. \frac{ \partial \Omega} { \partial \rho } \right|_{ \rho_{c},  T_{c} }   = \left. \frac{ \partial^2 \Omega} { \partial \rho^2 } \right|_{\rho_{c},  T_{c} } =0.
\end{align}
In Fig.~\ref{fig:nucleation1}(a) we report the bulk phase diagram obtained from the discretized grand free energy of  LJ fluid.
Solid curves depict the binodal, i.e. the locus of liquid-gas coexistence,
while dashed curves depict the spinodal, i.e. the boundary between the
metastable and the unstable regions.
The black circle designates the bulk critical point at $\rho_c \sim 0.3 $ and $T_c \sim 1.35$.

If we denote with $\rho_{v}$ the vapour coexistence density at a given
temperature, the supersaturation ratio is defined as $s= \rho / \rho_{v} $.
We will study the nucleation of vapour systems with identical temperatures,
but different initial supersaturation ratios.
Figure~\ref{fig:nucleation1}(b) depicts the free-energy landscape as a
function of the bulk density for such systems.
At coexistence $s=1$, two stable basins are present, which means that the system has equal probability of being in one of the two.
Increasing the supersaturation ratio enhances the stability of the liquid phase, thus leaving the vapour density in a metastable condition.
Also, the energy barrier that the system has to overcome to pass from the
vapour to the liquid phase decreases with $s$, until it becomes null at a
supersaturation corresponding to the spinodal line.
In such condition only one minimum of the grand free energy exists.

During a phase transition, the system will move from the initial uniform-density state to the final uniform-density state, while during the transition between the two states the density is non-uniform.
This means that the bulk grand-free energy in Fig.~\ref{fig:nucleation1}(b),
being only valid for uniform densities, describes the system in the
initial and final stages only, but it does not provide information on the
transition path.
The grand free energy for non-uniform systems is in general a function of
each cell density, i.e. it is an $n$-dimensional manifold.
To give a representative example of this, in Fig.~\ref{fig:nucleation1}(c) we
report our LJ grand-free energy for a non uniform system, constrained to have
only two varying densities $\widetilde{\rho}_1 = \left \{ {\rho}_1 = \dots
={\rho}_{n/2} \right \}$ and $\widetilde{\rho}_2 = \left \{ {\rho}_{n/2+1} =
\dots ={\rho}_{n} \right \} $.
The bulk free energy is then recovered for $\widetilde{\rho}_1 = \widetilde{\rho}_2$ (dotted black line).

Single trajectories of the vapour-to-liquid phase transition at different
supersaturation ratios, are reported in Fig.~\ref{fig:nucleation2}.
For comparison purposes, we perform simulations of the FDDFT and its mean-field (deterministic) counterpart.
In order for the transition to occur, the system grand free energy has to overcome an energy barrier.
Such passage requires a local injection of energy, thus it is triggered by
fluctuations.
As a consequence, the mean-field approach fails to describe the transition.
Moreover, as expected by looking at the energy barrier in Fig.~\ref{fig:nucleation1}, the transition is favoured by higher supersaturation ratios.

In addition to the presence of fluctuations, the phase transition is allowed due to the open boundary conditions imposed on the system.
These boundary conditions are described in subsection \ref{subsec:boundary}, and model the exchange of particles with a reservoir at constant temperature $T_{\text{res}}$ and chemical potential $\mu_{\text{res}}$.
The mass of the system can then increase (or decrease), thus permitting the transition from the lower-density minima in Fig.~\ref{fig:nucleation1} (b) to the higher-density ones.
However, it is important to remark that these boundary conditions do not simply add (or remove) mass to the system.
The imposed chemical potential at the boundary, $\mu_{\text{res}}$, can be iteratively solved to obtain the value of the density that satisfies it.
We choose $\mu_{\text{res}}$ so that this iterative algorithm may converge to one of the two minima in Fig.~\ref{fig:nucleation1} (b), depending on the initial conditions for the iteration.
For the two simulations in Fig.~\ref{fig:nucleation2} we always select to converge to the lower-density minimum in Fig.~\ref{fig:nucleation1} (b).
This is why, with identical boundary conditions, the mean-field deterministic simulation in Fig.~\ref{fig:nucleation2} remains at the lower-density minima in Fig.~\ref{fig:nucleation1} (b) and conserves the mass.
On the contrary, the FDDFT simulation in Fig.~\ref{fig:nucleation2} is able to increase the mass thanks to the constant density at the boundary, which allows a continuous exchange of particles.

The trend observed in Fig.~\ref{fig:nucleation2} is quantitatively analysed in Fig.~\ref{fig:nucleation3}, where we
report an ensemble average of $10$ nucleation
trajectories for each supersaturation ratio.
Figure~\ref{fig:nucleation3}(a) shows the average density increase as a
function of time.
The initial and final average system densities are consistent with the vapour and liquid bulk densities predicted by the grand-free energy analysis.

The free energy evaluated at each time as function of the average density is reported in Fig.~\ref{fig:nucleation3}(b).
The initial free-energy value, corresponding to the vapour metastable basin,
evolves in time in order to the reach the more stable liquid basin, as
predicted by Fig.~\ref{fig:nucleation1}(b).
It is interesting that the passage between the two basins implies a
slight increase in the free energy due to the energy barrier overcome by the
density field fluctuations.

Evidently,  the average density kinetics is characterized by three main
stages: 1) an initial latency period, 2) a growth period and, 3) an
asymptotic relaxation towards a plateau, corresponding to the liquid-phase
density.
This dynamics is consistent with the multi-stage nucleation pathway experimentally observed and theoretically studied in the phase-transition research community~\cite{DuranOlivencia2018}.
The growth period exhibits a linear-like trend, with slopes representing the nucleation growth rate $J$.
As reported in the plot in Fig.~\ref{fig:nucleation3}(c), an Arrhenius like
relation (as is the case with thermally activated processes) is observed
between $J$ and the grand-free energy barrier $\Delta \mathcal{E}$, i.e.
\begin{align}
J \sim K \exp{ - \frac{ \Delta \mathcal{E} }{T} },
\end{align}
where $J$ is the growth rate $K$ in the limit of a zero-energy barrier.
It is worth noticing that the pre-exponential factor $K$ in reality is not a constant, but can be often approximated as constant over limited supersaturation regions~\cite{lutsko2013classical,lutsko2015two}.

Finally, we remark that the finite-volume scheme is able to accurately
	simulate processes where the number of particles per cell is greater than 5,
	as showed in subsection \ref{subsub:std}. For any process that involves
	smaller scales one has to rely on MD simulations. This could be relevant for
	processes, such as nucleation, which may require capturing system
	features down almost to particle scales at initiation.

\section{Summary and conclusions}
\label{conclusions}

We have developed an efficient and robust finite-volume numerical scheme for solving stochastic gradient-flow equations. The scheme was exemplified with FDDFT and allows us to scrutinise the effects of thermal fluctuations on complex phenomena such as phase transitions.
Unlike previous numerical methodologies only applicable to a
limited range of free energies (e.g. ideal-gas free energies such as in
Refs~\cite{kim2017,voulgarakis2009bridging}), our proposed scheme deals
effectively with general free-energy functionals including
external fields or interacting potentials.

Our numerical methodology essentially comprises: a hybrid space
discretization based on central and upwind schemes, for both deterministic
and stochastic fluxes; a family of implicit-explicit Euler and Milsten
time integrators, together with a weak second-order Runge-Kutta scheme;
an adaptive time-step scheme, based on the Brownian bridge technique, which
ensures the non-negativity of the density; appropriate boundary
conditions.
What is more, the hybrid approach provides an optimal compromise between statistical properties of the stochastic field and spurious oscillations.
Additionally, the adaptive time-step feature of the scheme represents an alternative approach to preserving density positivity without including artificial limiters as in previous schemes.

The scheme is validated by means of several numerical applications.
First, we study the variance, temporal and spatial correlations, and
structure factor of an ideal gas at equilibrium, comparing the results of the
finite-volume solver with theoretical results from the literature and our own
MD simulations.
In agreement with previous works, we find that a minimum number of $5-10$
particles per cell is required in order for FDDFT to match atomistic
simulation results.
We the examine the out-of-equilibrium evolution of an ideal gas in a
double-well external potential.
Our stochastic solver accurately reproduces local
mean density, local density fluctuations and spatial correlations obtained
from MD simulations.
It should also be noted that for the deterministic case/DDFT where thermal
fluctuations are not included, the results are consistent with both FDDFT and
MD.
Finally, we simulate homogeneous nucleation kinetics of a fluid consisting of
particles interacting through an LJ-like potential.
Our results for the phase diagram match the theoretical results and serve so
as to illustrate the crucial role of fluctuations to surmount free-energy
barriers.
As expected, an exponential law is observed for the nucleation growth rate as
function of the metastable free-energy barrier.

\section*{Acknowledgments}
We gratefully acknowledge financial support from the Imperial College (IC)
Department of Chemical Engineering PhD Scholarship scheme, IC President's PhD
Scholarship scheme, ERC through Advanced Grant No. 247031 and and 883363 and EPSRC through
Grants No. EP/L027186, EP/L020564 and EP/P031587. The computations were
performed at the High Performance Computing center of IC.
Finally, we are grateful to the anonymous
	reviewers for insightful comments and suggestions.

\appendix

%
%
%
%

	\section{MD simulations details}

	MD simulations are performed using the Large-Scale Atomic/Molecular Massively
	Parallel Simulator (LAMMPS) \cite{plimpton1995fast}. Particle positions and
	velocities are integrated in time using the velocity-Verlet algorithm, with a
	time-step of $dt=0.001 \tau$. The system is simulated at constant temperature
	and volume, so that particle coordinates are consistent with the canonical
	ensemble (NVT). Specifically, the temperature $T=1$ is kept constant during
	the simulations using a Langevin thermostat. All the physical quantities are
	expressed in reduced units, i.e. they are nondimensionalized with the
	fundamental quantities $\sigma$, $\epsilon$ and $m$, representing distance,
	energy and mass, respectively. Further, without loss of generality, $\sigma$,
	$\epsilon$, $m$ and the Boltzmann constant $k_B$ are set equal to unity.
	
	As discussed extensively in \cite{russo2019macroscopic}, a macroscopic field
	$X(\mathbf{r},t)$ can be extracted from particle coordinates as
	$X(\mathbf{r},t)= \sum_{i} \chi_i  \phi ( \mathbf{r}_i(t) - \mathbf{r} ) $,
	where $\chi_i$ is the quantity of interest for particle $i$ at position
	$\mathbf{r}_i$ at time $t$, and $\phi$ is a kernel function (commonly a
	piecewise constant, Gaussian, or polynomial function). In this work, we adopt
	a piecewise constant function defined as:
	\begin{align}
	\phi (y) = \begin{cases}
	\frac{1}{ \Delta x } & \text{for } \left\Vert y\right\Vert< \Delta x/2 ,\\
	0 & \text{otherwise } ,
	\end{cases}
	\label{eq:piecewise_kernel}
	\end{align}
	where $\Delta x$ is the width of each bin. In each comparison, $\Delta x$ for MD simulations is taken to be the same with that for the discretized FDDFT.
	Using the above, the instantaneous macroscopic density profile for a single trajectory is computed as:
	\begin{align}
	\rho(x,t)= \sum_{i} m_i  \phi ( x_i (t) - x ) ,
	\end{align}
	where $m_i$ is the mass of the particle $i$.
	
	\paragraph{Equilibrium simulations}
	MD simulations of ideal gas fluids in equilibrium are performed using a fixed
	number of particles ($1,000$) in a 1D domain of length $2,000$ (in reduced
	units) with periodic boundary conditions. The system is equilibrated and then
	a run of $2 \times 10^{7}$ time steps is performed, during which fluid
	particle positions and velocities are stored every $10^{4}$ time steps for
	analysis. The process is repeated $10^{3}$ times to generate independent
	trajectories.
	
	\paragraph{Non-equilibrium simulations}
	
	MD simulations of ideal gas fluids in non-equilibrium conditions are
	performed using a fixed number of particles ($200$) in a 1D
	domain of length $200$ (in reduced units) with periodic boundary conditions,
	under an external potential:
	\begin{align}
	V(x)=5\left[\left(\frac{x}{200}\right)^4-\left(\frac{x}{200}\right)^2\right].
	\end{align}
	A run of $2 \times 10^{6}$ time steps is performed, during which fluid
	particle positions and velocities are stored every $10^{3}$ time steps for
	analysis. The process is repeated $10^{3}$ times with different (random)
	initial conditions to generate independent trajectories and gather
	statistics.

	\section{Time integrators stability analysis}
	
	\begin{figure}[t!]
			\begin{center}
				\includegraphics[scale=0.9]{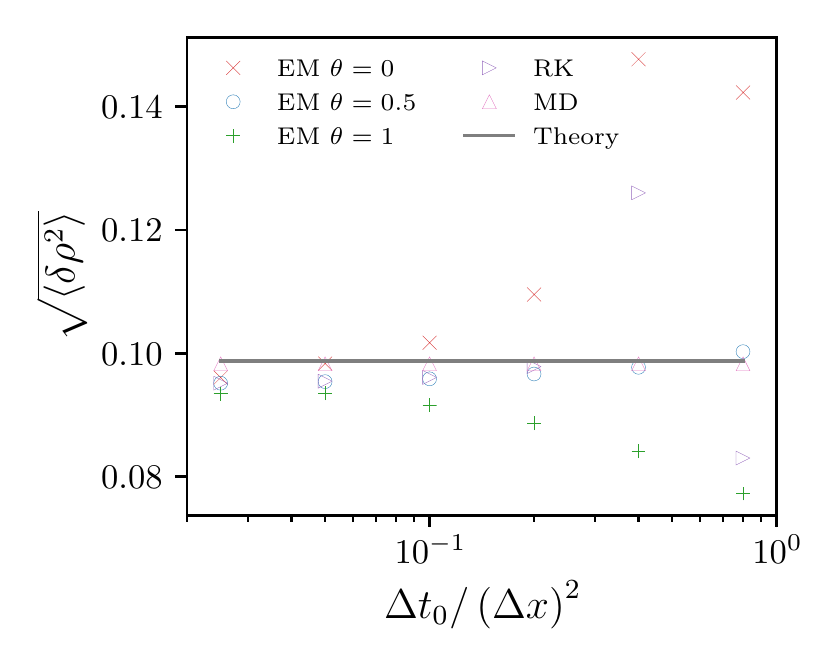}
				\protect\protect\caption{\label{fig:std_stability} Standard deviation $\sqrt{ \langle \delta \rho^2 \rangle}$ as a function of the normalized initial time step $\Delta t_0 / (\Delta x)^2$
					for an ideal gas in equilibrium. Temporal integrators. EM:
					Euler-Maruyama, MI: Milstein, RK: Runge-Kutta, MD: Molecular dynamics.
					Explicit ($\theta=0$), semi-implicit ($\theta=0.5$) and implicit
					($\theta=1$), Theory:
					Eq.~\eqref{eq:theoretstd}.}
			\end{center}
	\end{figure}
	
	Both stability and accuracy of the different time-integrators are relevant,
	given that large time-steps are required in many applications (for instance,
	for transitions occurring over long time-scales). In the main text, we
	focused on the accuracy of the schemes comparing finite-volume schemes, MD
	and theoretical results. Here we analyze the stability of the different time
	integrators with respect to the time-step size.
	
	Specifically, in Fig.~\ref{fig:std_stability} we report a comparison of the
	fluctuations' standard deviation obtained from selected time integrators
	and the MD-theoretical results for varying time step sizes $\Delta t_0$. Because of the adaptive time step adopted in the
	simulations,  the
	actual time step may not be constant throughout the simulations, and in fact it
	may be lower than $\Delta t_0$.
	The system considered here is the same ideal-gas system (with average density
	$\bar{\rho}=0.5$) used for the analyses in the main text. The cell size
	adopted is $\Delta x =50$, corresponding to a number of particles per cell, $N_c=25$.
	We do not report the results for the Milstein schemes, as in
	several tests we did not observe any relevant difference between the
	Milstein scheme and the Euler-Maruyama one as far as the mean, variance and
	correlations are concerned.
	Figure~\ref{fig:std_stability} shows that the semi-implicit scheme
	outperforms both explicit and implicit schemes at high $\Delta t_0/ \left(
	\Delta x \right)^2$, becoming the time-integrator of choice for computations
	requiring large time steps. Moreover, the explicit Runge-Kutta scheme shows enhanced stability compared to both implicit and explicit Euler-Maruyama.

\bibliographystyle{siam}
\bibliography{references}

\end{document}